\documentclass[draft,11pt]{article}

\usepackage{lineno,hyperref}
\modulolinenumbers[5]

\usepackage{geometry}
\usepackage{pgfplots}
\usepgfplotslibrary{groupplots}
\usepackage{color}

\usepackage{amsmath,amssymb,amsthm}
\theoremstyle{plain}
\newtheorem{theorem}{Theorem}
\newtheorem{lemma}{Lemma}
\newtheorem{proposition}{Proposition}

\theoremstyle{definition}
\newtheorem{definition}{Definition}
\newtheorem{example}{Example}
\newtheorem{remark}{Remark}
\usepackage[ruled,vlined,linesnumbered]{algorithm2e}
\SetKw{Continue}{continue}
\usepackage{caption} 
\usepackage{subcaption}
\usepackage{xspace}
\usepackage{adjustbox}
\usepackage{textcomp}
\usepackage{mathtools}
\usepackage{multirow}

\usepackage{tikz}
\usetikzlibrary{automata,arrows,shapes}
\usepackage{mathrsfs}

\pgfdeclarelayer{background}
\pgfsetlayers{main,background}
\tikzstyle{vertex}=[circle,fill=black!10,minimum size=20pt,inner sep=0pt]
\tikzstyle{edge} = [draw,thick]
\tikzstyle{weight} = [font=\small]
\tikzstyle{redundant edge} = [draw,line width=5pt,-,red,opacity=.7]

\usepackage{pgfplots}
\pgfplotsset{width=.3\linewidth,compat=1.9}

\DeclarePairedDelimiterX\Set[1]\{\}{%
  
  #1 }

\newcommand{\length}{\mathsf{length}}
\newcommand{\distance}{\mathsf{dist}}

\newcommand{\la}{\langle}
\newcommand{\ra}{\rangle}

\newcommand{\norm}[1]{\lVert#1\rVert}

\newcommand{\algorithmfootnote}[2][\footnotesize]{%
  \let\old@algocf@finish\@algocf@finish
  \def\@algocf@finish{\old@algocf@finish
    \leavevmode\rlap{\begin{minipage}{\linewidth}
    #1#2
    \end{minipage}}%
  }%
}

\AtBeginDocument{}
\setkeys{Gin}{draft=false}      
\usepackage[obeyDraft,colorinlistoftodos]{todonotes}

\colorlet{SL}{red!20!yellow!60!white}
\colorlet{SLline}{SL!80!black}

\begin{document}


\title{Semi-dynamic shortest-path tree algorithms for directed graphs with arbitrary weights 
}

\author{
Sanjiang Li$^{1}$ and Yongming Li$^{2}$\\
\small $^1$\,Centre for Quantum Software and Information,  \\
\small  University of Technology Sydney, Sydney, Australia\\
\small $^2$\,School of Computer Science,
Shaanxi Normal University, Xi'an, China}

\maketitle

\begin{abstract}
\noindent Given  a directed graph $G$ with arbitrary real-valued weights, the single source shortest-path problem (SSSP) asks for, given a source $s$ in $G$, finding a shortest path from $s$ to each vertex $v$ in $G$. A classical SSSP algorithm detects a negative cycle of $G$ or constructs a shortest-path tree (SPT) rooted at $s$ in $O(mn)$ time, where $m,n$ are the numbers of edges  and vertices in $G$ respectively. In many practical applications, new constraints come from time to time and we need to update the SPT frequently.  Given an SPT $T$ of $G$, suppose the weight on certain edge is modified. We show by rigorous proof that the well-known {\sf Ball-String} algorithm for positive weighted graphs can be adapted to solve the dynamic SPT problem for directed graphs with arbitrary weights. 
Let $n_0$ be the number of vertices that are affected (i.e., vertices that  have different distances from $s$ or different parents in the input and output SPTs) and $m_0$  the number of edges incident to an  affected vertex. The adapted algorithms terminate in $O(m_0+n_0 \log n_0)$ time, either detecting a negative cycle (only in the decremental case) or constructing a new SPT $T'$ for the updated graph. We show by an example that the output SPT $T'$ may have more than necessary edge changes to $T$. To remedy this, we give a general method for transforming $T'$ into an SPT with minimal edge changes in time $O(n_0)$ provided that $G$ has no cycles with zero length.
\end{abstract}



\section{Introduction}
The problem of finding shortest paths is a fundamental problem in computer science with numerous applications, e.g., in Internet routing protocols \cite{NarvaezST01}, network optimisation \cite{fredman1987fibonacci}, and temporal reasoning \cite{DechterMP91,Kong+18}. Let $G=(V,E)$ be a directed graph in which each edge $(u,v)\in E$ has an arbitrary real-valued (possibly negative) weight $\omega(u,v)$. A path $\pi$ in $G$ is a sequence of vertices $\la v_1, v_2, ..., v_k\ra$ such that,  $(v_i,v_{i+1})$ is an edge in $E$ for each $1\leq i < k$. The \emph{length} of $\pi$, denoted by $\length_G(\pi)$, is the sum of  the weights of all these edges. We call path $\pi$ a \emph{shortest path} from $x=v_1$ to $y=v_k$ if there is no path $\pi'$ from $x$ to $y$ which has a shorter length than $\pi$. In this case, we also call $\length_{G}(\pi)$ the \emph{distance} from $x$ to $y$. The \emph{single source shortest path} problem (SSSP) is the problem of, given a source $s$ in $G$, finding a shortest path from $s$ to $v$ for every vertex $v$. Analogously, we have the \emph{single destination shortest path} problem, which can be solved as an SSSP if we reverse the direction of each edge in $G$. The \emph{all pairs shortest path} problem (APSP) is the problem of finding  a shortest path from $u$ to $v$ for any pair $(u,v)$ of vertices in $G$, which can be solved by running SSSP for each vertex $s$ in $G$.

In this paper, we are mainly concerned with SSSP, which has been studied for more than 60 years but remains an active topic in computer science. The set of shortest paths for a single source $s$ can be compactly represented in a \emph{shortest-path tree} (SPT) $T$ rooted at $s$, which is a spanning tree of $G=(V,E)$ such that every tree path is a shortest path. Note that if negative weights exist in $G$ it is likely that  there is no shortest path for some vertex pairs. Thus, an SSSP algorithm for this general case either detects a negative cycle of $G$ or constructs an SPT rooted at $s$. Many algorithms have been devised in the literature. Two most well-known are the one devised by Dijkstra \cite{Dijkstra59}, which has worst-case time complexity  $O(m + n\log n)$ with Fibonacci heap \cite{fredman1987fibonacci}, and the Bellman-Ford algorithm \cite{Bellman58,Moore59}, which has worst-case time complexity $O(nm)$, where $m=|E|$ and $n=|V|$ are, respectively, the number of edges  and vertices in $G$. While Dijkstra's algorithm only applies to directed graphs with non-negative weights, Edmonds and Karp \cite{EdmondsK72} introduced a technique that can transform \emph{some} directed graphs with arbitrary weights to a non-negative weighted directed graph with the same shortest paths. Indeed, suppose we have a function $f: V \to \mathbb{R}$ such that $\omega_f(u,v) \equiv f(u)+\omega(u,v) - f(v) \geq 0$ for any edge $(u,v)\in E$ and write $G_f$ for the directed graph with the weight $\omega(u,v)$ replaced by $\omega_f(u,v)$ for each edge $(u,v)\in E$. Then $G$ and $G_f$ have the same set of shortest paths.

In many practical applications, new constraints come from time to time and we need to update the SPT frequently.  More precisely, an existing  edge may be deleted, or its weight may be updated (i.e., increased or decreased), or a new edge with real weight may be inserted. In what follows, we only consider weight updates, as edge deletion and insertion may be regarded as special cases of edge weight updates if we allow infinite edge weight. Assume that we already have constructed, using some SSSP algorithm, an SPT $T$ for $G$, and suppose the weight of an edge is updated.
To generate a new SPT of the updated graph, without doubt, we could run the SSSP algorithm again. This is, however, usually inefficient and unnecessary,  especially when we want to keep the topology of the existing SPT as much as possible in applications like Internet routing. We then ask: if we can generate a new SPT without running any static SSSP algorithm from scratch and preserve as much as possible the existing tree structure?

This problem, known as (one-shot) dynamic SSSP,  has been studied as early as 1975 \cite{SpiraP75} and attracted attention of researchers from various research areas, including networks, theoretical computer science, and artificial intelligence, see e.g.  \cite{Aioanei12,BernsteinC16,Buriol08,ChanY09, EvenS81,FranciosaFG97,FrigioniMN94,FrigioniMN96,FrigioniMN98,FrigioniMN00, Madkour+17,McQuillanRR95, NarvaezST00, NarvaezST01, PlankenWK12,RodittyZ11,  Xiao+07}. Many algorithms have been proposed to solve this problem.  In this paper, we call such an algorithm \emph{fully dynamic} if it can deal with both edge weight increase and decrease; and call it \emph{semi-dynamic} if it can deal with either edge weight increase or decrease but not both. A semi-dynamic SPT algorithm is called \emph{incremental}  (\emph{decremental}, resp.) if it can only deal with edge weight increase (decrease, resp.). 
Most existing (semi-)dynamic SPT algorithms are dynamic variants of the static Dijkstra's algorithm, and thus cannot be directly applied to solve the dynamic SPT problem for directed graphs with negative weights. 
Narv{\'{a}}ez et al. \cite{NarvaezST01} propose an elegant dynamic SPT algorithm, called {\sf Ball-String} henceforth, which always selects and extends the edge that leads to the minimum increase (or the maximum decrease) in path length. Compared with existing algorithms, it is more efficient and only makes much fewer edge changes to the existing SPT structure as it always tends to consolidate vertices in the whole branch instead of only one vertex.

The original {\sf Ball-String} algorithm \cite{NarvaezST01} only considers directed graphs with positive weights. In the literature, Ramalingam and Reps \cite{Ramalingam96,RamaR96} are the first authors who have considered the dynamic SSSP problem for directed graphs with arbitrary weights. Their idea is to first design Dijkstra-like semi-dynamic algorithms for directed graphs with positive weights and then solve the dynamic SSSP problem for directed graphs with arbitrary weights by using the technique introduced by Edmonds and Karp. Indeed, the original distance function $\distance_G: V\to \mathbb{R}$ is used to act as the translation function $f$, where $\distance_G(v)$ is the distance from $s$ to $v$ in $G$.  For directed graphs with nonzero-length cycles, their algorithms terminate in $O(\norm{\delta}+ |\delta| \log |\delta|)$ time per update, where $|\delta|$ is the number of vertices affected by the edge weight change $\delta$, and $\norm{\delta}$ is the number of affected vertices plus the number of edges incident to an affected vertex. When zero-length cycles present, they show that there are no semi-dynamic algorithms that are \emph{bounded} in terms of the output changes $|\delta|$ and $\norm{\delta}$.

To deal with  directed graphs with zero-length cycles, Frigioni et al. \cite{Frigioni1998fully,FrigioniMN03} propose new semi-dynamic SSSP algorithms with worst-time complexity $O(m\log n)$ per update, which, like \cite{NarvaezST01}, also use the minimum increase or maximum decrease as the search criterion to determine in which direction the change should be propagated. Demetrescu et al. \cite{demetrescu2000maintaining} present the first experimental study of the fully dynamic SSSP problem in directed graphs with arbitrary weights. It was shown there that all the considered dynamic algorithms (including that of Ramalingam and Reps) are faster by several orders of magnitude than recomputing from scratch with the best static algorithm. Similar experimental results were also reported in \cite{Buriol08}, where Buriol et al. introduced a technique to reduce the size of the priority queue used in a dynamical SSSP algorithm. The idea was actually also used in {\sf Ball-String} \cite{NarvaezST01}, which enqueues a vertex only if a better alternative path through it is found.   Inspired by the {\sf Ball-String} algorithm, Rai and Agarwal \cite{RaiA11}
present semi-dynamic algorithms for maintaining SPT in a directed graph with arbitrary weights. However, neither proof nor theoretical analysis were provided for the correctness and the computational complexity of their algorithms, which we believe are less efficient as they put every affected vertex in the priority queue and consolidate only one vertex in each iteration.

In this paper, we consider dynamic directed graphs with possibly negative weights and present two efficient semi-dynamic algorithms based on the {\sf Ball-String} algorithm. Our adapted algorithms either detect a negative cycle (only in the decremental case)  or generate a new SPT $T'$ in $O(m_0+n_0\log n_0)$ time, where $n_0$ is the number of affected vertices that have different distances from the source vertex or different parents in $T'$ and $T$ and $m_0$ is the number of edges that are incident to an affected vertex. Compared with running the best static SSSP algorithm (e.g., the  $O(mn)$ Bellman-Ford algorithm) from scratch or running the $O(m\log n)$ algorithms of Frigioni et al. \cite{Frigioni1998fully,FrigioniMN03}, this is more efficient. Moreover, unlike the algorithms of Ramalingam and Reps \cite{Ramalingam96,RamaR96}, our algorithms can deal with directed graphs with zero-length cycles and are more \emph{efficient} because they extract fewer vertices from the priority queue and always consolidate the whole branch when a vertex is extracted. We further note that, in   \cite{RamaR96} and \cite{FrigioniMN03}, the aim is to maintain \emph{all} shortest paths from the source vertex $s$, instead of a shortest-path tree from $s$, which contains only one shortest path from $s$ to any vertex. 

It is also worth noting that the correctness and the minimality of the output SPT of the adapted {\sf Ball-String} algorithms cannot be guaranteed by simply generalising the proof given for directed graphs with only positive weights in \cite{NarvaezST01} (cf. Remark~\ref{remark}  in page~\pageref{remark} of this paper for more details). 
When negative weights present, the SPT $T'$ output by the adapted {\sf Ball-String} algorithms  may have more than necessary edge changes.  Another contribution of this paper is a general method for transforming $T'$ into an SPT with minimal edge changes in time $O(n_0)$  provided that $G$ has no zero-length cycles.

The remainder of this paper is organised as follows: after some preliminaries and a brief introduction of the original {\sf Ball-String} algorithm in Section 2, we present the incremental and decremental SPT algorithms for directed graphs with arbitrary weights in, respectively, Sections 3 and  4, where rigorous proofs of the correctness and finer theoretical analysis of the output complexity of these algorithms are also given. The last section concludes the work.
  
\section{Preliminaries and related work}

Let $G=(V,E)$ be a weighted directed graph.  For any edge $e=(u,v)$ in $E$, we call $u,v$, respectively, the \emph{tail} and the \emph{head} of $e$. The weight of an edge $e=(u,v)$ in $E$ is written as $\omega({u,v})$. In applications, these weights may denote distances between two cities or costs between two routers and thus are usually non-negative. The simple temporal problem  \cite{DechterMP91} can also be represented as a weighted directed graph, where each vertex $x$ represents a time point $t_x$, and a weight $\omega(u,v)$ on the edge from vertex $u$ to $v$  specifies that the time difference of $u$ and $v$ is upper bounded by $\omega(u,v)$, i.e., $t_v-t_u \leq \omega(u,v)$. Clearly, weights in this case may be negative. 

In what follows, we assume that $s$ is a fixed vertex in $G$ and every vertex $v$ in $G$ is \emph{reachable} from $s$, i.e., there is a path from $s$ to $v$. We say a path $\pi=\la v_1, v_2, ..., v_k\ra$ in $G$ is a \emph{cycle} if $v_1=v_k$. A path is called \emph{simple} if it contains no cycle. A cycle in $G$ is \emph{negative} if it has negative length. Similarly, a 0-cycle is a cycle with zero length.

For any two vertices $x$ and $y$, if there is a path from $x$ to $y$ in $G$, then there is a shortest path from $x$  to $y$ if and only if no path from $x$  to $y$  contains a negative cycle, and, if there is any shortest path from $s$  to $t$, there is one that is simple (cf. \cite{Tarjan83}).

\begin{figure}[htbp]
\centering
\begin{tabular}{ccc}
\includegraphics[width=.28\textwidth]{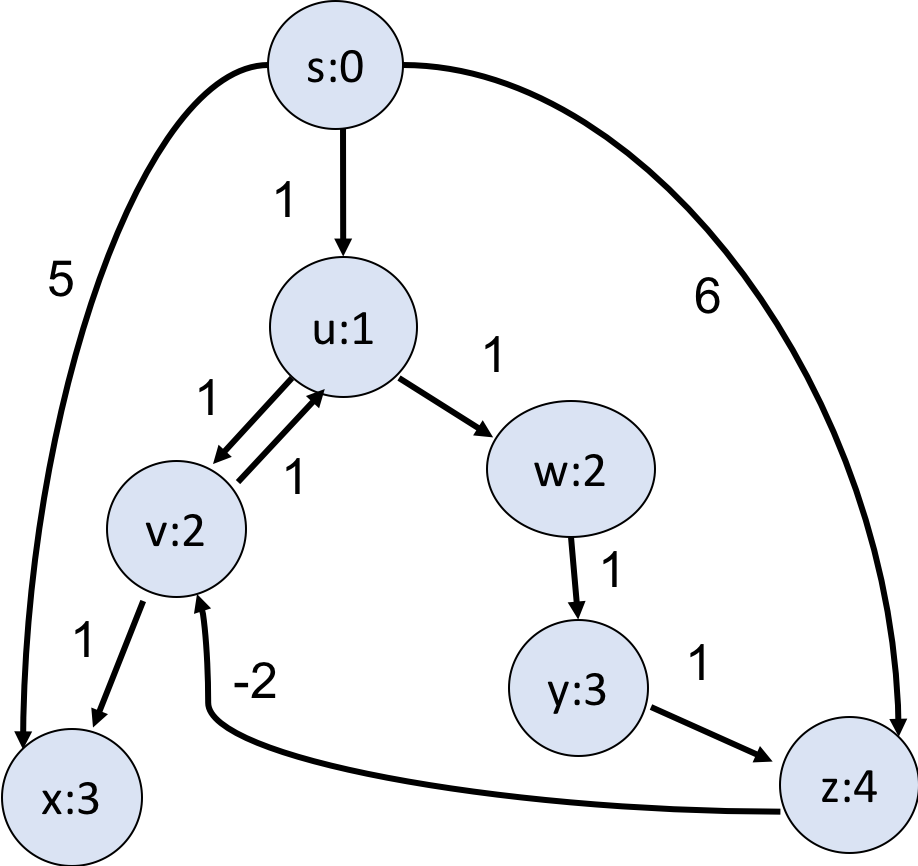}
&& 
\includegraphics[width=.28\textwidth]{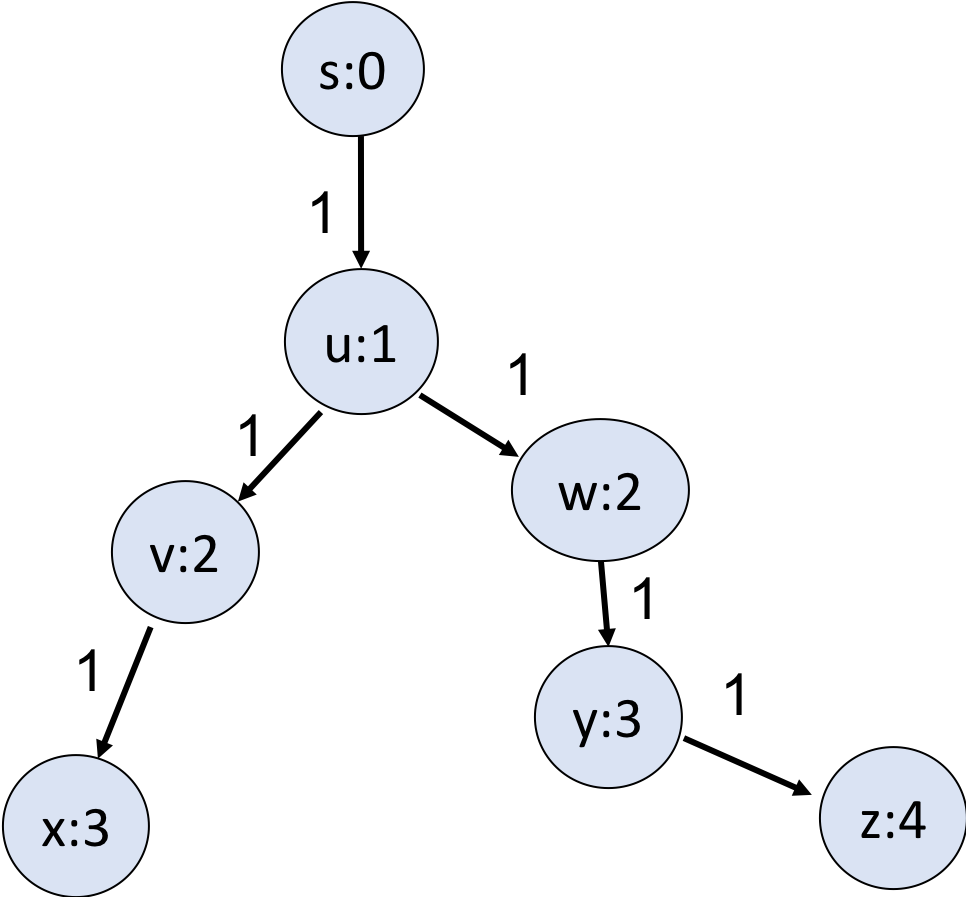}
\end{tabular}
\caption{A directed graph with negative weights (left) and its SPT (right), where we write, for example, $u:1$ to denote that the distance of vertex $u$ is 1.}  
\label{fig:digraph}
\end{figure}
A shortest-path tree (SPT) of $G$ is a spanning tree $T$ rooted at $s$  such that each path of $T$ from $s$ is a shortest path in $G$ (see Figure~\ref{fig:digraph} (right) for an illustration). SPT is a compact way to encode shortest path information for a single source $s$, and $G$ contains shortest paths from $s$ to all other vertices if and only if $G$ contains an SPT rooted at $s$ \cite{Tarjan83}. Clearly, an STP does not encode every shortest path from $s$. For example, in Figure~\ref{fig:digraph}, $\pi_1=\la s,u,v\ra$ and $\pi_2=\la s,u,w,y,z,v\ra$ are two shortest paths from $s$ to $v$ and only $\pi_1$ is encoded in the SPT. 

\subsection{Our task}

Let $G=(V, E)$ be a directed graph whose edges have real-valued (possibly negative) weights $\omega(u,v)$ for each $(u,v)\in E$. Suppose $T$  is an SPT of $G$ rooted at $s \in V$  and $e_0=(x_0, y_0) \in E$ is an edge in $G$ whose weight $\omega(x_0,y_0)$ has been changed to $\omega'(x_0,y_0)$.  Let $G'$ be the directed graph obtained by replacing  $\omega(x_0,y_0)$ on $(x_0,y_0)$ with $\omega'(x_0,y_0)$. 

Our task is to construct a  shortest-path tree $T'$ of $G'$ without running the SSSP algorithm from scratch. Note that there are two cases: one is weight increment, the other is weight decrease.  In the latter case, it is possible that the updated weighted graph has a negative cycle. For example, if the weight of edge $(v,u)$ in Fig.~\ref{fig:digraph} is decreased from 1 to -2, then the cycle $\pi=(u,w,y,z,v,u)$ has revised length -1. In what follows, we say a directed graph is \emph{inconsistent} if it has a negative cycle.

\subsection{Notations}
For $u,v\in V$, we write $\distance_G{(u,v)}$ for the distance from $u$ to $v$ in $G$, and write $\distance_{G'}{(u,v)}$ for the distance from $u$ to $v$ in $G'$ if it exists. Since $s$ is the fixed source, we also write $\distance_G(v)$ and $\distance_{G'}(v)$ for $\distance_G{(s,v)}$ and $\distance_{G'}{(s,v)}$, respectively,  if $s$ is known. For a path $\pi= \la u_0, u_1, ..., u_i, ..., u_j, ..., u_k \ra$, we write $\pi[u_i, u_j]$ for the subpath $\la u_i, ..., u_j \ra$. The length of $\pi$ in $G$, written as $\length_G(\pi)$, is defined as the sum $\omega_{u_0,u_1} + \omega_{u_1,u_2} + ... + \omega_{u_{k-1},u_k}$. Analogously, we define $\length_{G'}(\pi)$, the length of $\pi$ in $G'$. Given two paths $\pi_1= \la u_0, u_1, ..., u_k \ra$ and $\pi_2= \la u'_0, u'_1, ..., u'_{k'} \ra$, if $u_k= u'_0$, then we write $\pi_1+\pi_2$ for their concatenation $\la u_0, u_1, ..., u_k=u'_0, ..., u'_{k'} \ra$. In particular, if $\pi_2=\la u'_0,u'_1\ra$ is an edge, the we write $\pi_1+(u'_0,u'_1)$ for the concatenation.   

The following result is simple but useful.
\begin{lemma}\label{lem:len_paths}
Let $\pi$ be a path in $G$. Then $\length_{G'}(\pi) = \length_G(\pi)$ if $e_0=(x_0,y_0)$ is not in $\pi$ and $\length_{G'}(\pi) = \length_G(\pi) + \omega'(x_0,y_0) - \omega(x_0,y_0)$ if $e_0$ is in $\pi$.
\end{lemma}

Suppose $T$ is the fixed SPT of $G$ before updating. For any vertex $v\in V$, we write $d_T(v)$ and $\pi_{s,v}$   for, respectively, the distance and the path from $s$ to $v$ in $T$ and write $\downarrow\!{v}$ ($\uparrow\!{v}$, resp.) for the set of descendants (ancestors, resp.) of $v$ in $T$. In this paper we require that $v$ is contained in both $\downarrow\!{v}$ and $\uparrow\!{v}$.

\subsection{The {\sf Ball-String} algorithm }

The {\sf Ball-String} algorithm was first described in \cite{NarvaezST01} for dynamic graphs with only positive weights. When negative weights are allowed and the weight of an edge is decreased, the algorithm cannot be directly applied as there may exist negative cycles in the updated graph. We thus need to introduce procedures for checking negative cycles. Moreover, as we will show in section 3, the algorithm does not always output an SPT with minimal edge changes even in the incremental case. 

Suppose  $e_0 = (x_0, y_0)$ is an edge in $G$ whose weight $\omega(x_0,y_0)$ has been changed to $\omega'(x_0,y_0)$. Given an SPT $T$ of $G$, the {\sf Ball-String} algorithm \cite{NarvaezST01} constructs an SPT $T'$ of $G'$ step by step from $T$: in each step it changes the parent of only one vertex. Recall that Dijkstra's algorithm maintains a priority queue (i.e., a heap) $Q$ of vertices, each vertex in $Q$ associated with a tentative distance,  and, in each step, it extracts from $Q$ the vertex with the smallest tentative distance. Similarly, the {\sf Ball-String} algorithm also maintains a priority queue $Q$ of vertices, but here each vertex $y$ is associated with a pair $(x, \Delta)$, where $x$ is the tentative parent of $y$ and $\Delta$ is the difference of the lengths of the tentative path $\widehat{\pi}_{s,x}+(x,y)$  and the baseline path $\pi_{s,y}$. Note here we write $\widehat{\pi}_{s,x}$ and $\pi_{s,y}$ for, respectively,  the path from $s$ to $x$ in the current spanning tree and the path from $s$ to $y$ in $T$. In each iteration, it extracts from $Q$ the vertex $y$ which has the smallest increase (or the largest decrease), updates the parent of $y$ as $x$, consolidates the descendants of $y$  (including $y$ itself) in the current spanning tree, and adds  in $Q$ new vertices which have tentative parents that have just been consolidated.


We first formalise the notion of the $\Delta$ value, which was implicitly described in \cite{NarvaezST01}.

\begin{definition}
\label{dfn:dominate}
Given a weighted directed graph $G=(V,E)$, let $T$ be an SPT of $G$ at root $s\in V$. Suppose $\omega(x_0,y_0)$ has been changed to $\omega'(x_0,y_0)$ for some $e_0=(x_0,y_0)$. Let $\widehat{T}$ be the current tree constructed step by step from  $T$ by changing the parent of only one vertex in each step. For any edge $e=(x,y)$ in $E$, we define 
\begin{align} \label{eq:Delta(e)}
\Delta_{\widehat{T}}(e) = d_{\widehat{T}}(x) + \omega^*(e) - \distance_G(y),
\end{align} 
where $\omega^*(e_0) = \omega'(e_0)$ and $\omega^*(e) = \omega(e)$ if $e\not=e_0$, $d_{\widehat{T}}(x)$  is the distance from root $s$ to $x$ in the current tree $\widehat{T}$, and $\distance_G(y)$ is the distance from $s$ to $y$ in $G$.
\end{definition}
Suppose $(x,y)\in E$ and $x$ is a consolidated vertex but $y$ is not. Here a vertex $x$ is consolidated if a shortest path from $s$ to $x$, as well the real distance of $x$, in $G'$ has been found in previous iterations and that will not change in the following iterations.  Let $\widehat{\pi}_{s,x}$ be the path from $s$ to $x$ in the current spanning tree $\widehat{T}$ and $\pi_{s,y}$ the path from $s$ to $y$ in $T$. We select $\pi_{s,y}$ as the \emph{baseline path} and compare if $\widehat{\pi}_{s,x}+(x,y)$ is better than $\pi_{s,y}$ in $G'$. By Lemma~\ref{lem:len_paths},
in the incremental case, $\pi_{s,y}$ has length $\distance_G(y) + \omega'(x_0,y_0) - \omega(x_0,y_0)$ in $G'$, provided that $e_0$ is an edge in $T$ and $y$ is a descendant of $y_0$; in the decremental case, it has length $\distance_G(y)$ in $G'$, provided that $y$ is not a descendant of $y_0$.
If the path $\widehat{\pi}_{s,x}+(x,y)$ is better than $\pi_{s,y}$ in $G'$, then we put $\la y, (x,\Delta_{\widehat{T}}(e)\ra$ in $Q$ if $y$ is not there or update the $\Delta$ value of $y$ in $Q$ by $\Delta_{\widehat{T}}(e)$ if it is smaller.  In each iteration, the algorithm selects the vertex $y$ with the smallest $\Delta$ value to extract. When there are two vertices with the same smallest $\Delta$ value, it selects the one which has a smaller current distance from $s$. If more than one vertex has the same smallest $\Delta$ value and the smallest distance from $s$, then it selects the one which is closer to the root. If there are still ties, pick any vertex to extract. 

Two different vertices in the queue $Q$ may be added or modified in different iterations. Suppose $y$ is modified in iteration $i$ but not later. Let $T_i$ and $T_j$ ($j>i$) be the spanning tree constructed in iterations $i$ and $j$ respectively. As will become clear in sections 3 and 4, the paths from $s$ to $y$ in $T_i$ and $T_j$ are the same. The following lemma guarantees that the $\Delta$ value of $y$ does not change even if we reevaluate it using $T_{j}$ instead of $T_i$.  

\begin{lemma} \label{lem:2trees}
Suppose $T'$, $T''$ are two spanning trees of $G'$ and $(x,y)$ is an edge in $G'$. If the paths from $s$ to $x$ in $T'$ and $T''$ are identical, then $d_{T'}(x) = d_{T''}(x)$ and $\Delta_{T'}(e) = \Delta_{T''}(e)$, where $d_{T'}(x)$ and $d_{T''}(x)$ are, respectively, the distances from $s$ to $x$ in $T'$ and $T''$.
\end{lemma}

Suppose $y$ is associated with the pair $(x,\Delta)$ when it is enqueued in (extracted from) $Q$. For convenience,  we say interchangeably that the edge $(x,y)$ is enqueued (extracted).

\section{The incremental case}
Given a weighted directed graph $G=(V,E)$, let $T$ be an SPT of $G$ at root $s\in V$. Suppose $e_0=(x_0,y_0)$ is an edge in $G$ and its weight has been increased from $\omega(e_0)$ to $\omega'(e_0)$. Write $\theta = \omega'(e_0) - \omega(e_0)$. Then $\theta>0$.  Let $G'=(V,E')$ be the graph obtained by replacing the weight of $e_0$ with $\omega'(e_0)$.  

\subsection{Description of the algorithm}
We present Algorithm~\ref{alg:+} for constructing a new SPT of $G'$, which is in essence the {\sf Ball-String} algorithm of \cite{NarvaezST01} for the incremental case.

We first observe the following simple result for the incremental case.
\begin{lemma}\label{lem3}
Suppose $e_0=(x_0,y_0)$ is an edge in $G$ and its weight has been increased from $\omega(e_0)$ to $\omega'(e_0)$.
For any path $\pi$  in $G$, we have $\length_{G}(\pi) \leq \length_{G'}(\pi) \leq \length_{G}(\pi) + \theta$, where $\theta=\omega'(e_0) - \omega(e_0)$. In particular, $G'$ has no negative cycle. 
If $e_0$ is not in $T$, then $T$ remains an SPT of $G'$. If $e_0$ is in $T$, then, for each vertex $v$ that is not an descendant of $y_0$, the path $\pi_{s,v}$ from $s$ to $v$ in $T$ remains a shortest path. 
\end{lemma}

By Lemma~\ref{lem3}, we only need to consider the case when $e_0$ is in $T$. In this case,  for any vertex $x$ outside $\downarrow\!{y_0}$, the path $\pi_{s,v}$ from $s$ to $v$ in $T$ remains a shortest path. Therefore, only vertices in $\downarrow\!{y_0}$ need to be examined.

In the initialisation iteration (lines~\ref{line:3+}-\ref{line:6+}), we write $X_0= V\setminus \downarrow\!{y_0}$, $Y_0=\ \downarrow\!{y_0}$ and set $Q=\varnothing$. We then check (lines~\ref{line:7+}-\ref{line:10+}) if the tree needs modification by examining all edges from $X_0$ to $Y_0$: if any such edge $(x,y)\not=e_0$ leads to a path $\pi \equiv \widehat{\pi}_{s,x}+(x,y)$ that is better than the baseline path $\pi_{s,y}$, then we put $\la y, (x, \Delta_{T}(x,y))\ra$ in $Q$, where $\widehat{\pi}_{s,x}$ is the path from $s$ to $x$ in the current spanning tree (viz., $T$), and 
$$\Delta_{T}(x,y) = \distance_T(x) + \omega(x,y) - \distance_G(y) =  \length_{G'}(\pi) - \length_{G}(\pi_{s,y}).$$ 
Because $\length_{G'}(\pi)=\length_{G}(\pi)\!=\!\distance_G(x)+\omega(x,y)$ and $\length_{G'}(\pi_{s,y})=\length_{G}(\pi_{s,y})$ $+\theta=\distance_G(y)+\theta$ (by Lemma~\ref{lem:len_paths}), we have $$\length_{G'}(\pi) < \length_{G'}(\pi_{s,y})\ \mbox{iff}\ \distance_G(x)+\omega(x,y)<\distance_G(y)+\theta\ \mbox{iff}\ \Delta_{T}(x,y)< \theta.$$ If there is no edge $(x,y)$ from $X_0$ to $Y_0$ satisfying $\Delta_{T}(x,y)<\theta$, then $Q$ is empty and we declare that $T$ remains an SPT of $G'$ and return $T$ as the output. 

Suppose $Q$ is nonempty. In the while loop (lines~\ref{line:11+}-\ref{line:24+}), the algorithm modifies $T_0=T$ step by step and, in each iteration, select  one vertex $\langle y, (x,\Delta)\rangle$ (or, equivalently, select the edge $(x,y)$) with the minimum $\Delta$  from $Q$ and replace the parent of $y$ as $x$ in the current spanning tree. Suppose edge $e_i=(x_i,y_i)$ is selected from $Q$ in the $i$-th iteration. Let $T_{i-1}$ be the spanning tree  of $G'$ constructed in the $(i-1)$-th iteration. We construct $T_{i}$ by changing the parent of $y_{i}$ as $x_{i}$ from its predecessor $T_{i-1}$. Then, we set $X$ as the descendants of $y_i$ in $T_i$ and consolidate all vertices in $X$ in this iteration. Let $Y$ be the set of vertices that remain unsettled. 
For each edge $e=(x,y)$ in $G$ with $x \in X$ and $y \in Y$, we define 
$$\Delta_{T_{i}}(e) = d_{T_{i}}(x) + \omega(e) - \distance_G(y)$$ 
as in \eqref{eq:Delta(e)}, where $d_{T_{i}}(x)$ represents the length of the path from $s$ to $x$ in the current spanning tree $T_{i}$, which, as we shall prove in Proposition~\ref{prop:Vi}, is a shortest path in $G'$. We then examine if we can extend this shortest path to $y$ through $e$. As before, the value $\Delta_{T_{i}}(e)$ represents the difference between such a new candidate path via $e$ to the original shortest path $\pi_{s,y}$ in $T$. Apparently, the smaller $\Delta_{T_{i}}(e)$ is the better. If $\Delta_{T_{i}}(e)\geq \theta$, then it is impossible to get a path via $e$ that is better than the baseline path $\pi_{s,y}$. In case $\Delta_{T_{i}}(e) < \theta$, there are two subcases: if  $y$ is not in $Q$, then we put $\la y, (x, \Delta_{T_{i}}(e)) \ra$ in $Q$; if $y$ is in $Q$, then we update the $\Delta$ value of $y$ in $Q$ with $\Delta_{T_{i}}(e)$ if the latter is smaller. 

The algorithm continues in this way until $Q$ becomes empty. Note that, immediately after $e_i=(x_i,y_i)$ is extracted, $y_i$ and all its descendants in the current tree are consolidated (line~\ref{line:17+}) and removed (line~\ref{line:18+}) from the queue $Q$ and $Y$ (the set of unsettled vertices)   if they are there. Since lines~\ref{line:19+}-\ref{line:24+} only examine edges from newly consolidated  vertices to unsettled ones, $y_i$, as well as any other consolidated vertex, will not be put in $Q$ again. As a consequence, the algorithm will stop in at most $n-1$ iterations, where $n$ is the number of vertices in $G$.

The most important advantage of the {\sf Ball-String} algorithm lies in that,  in each iteration, it consolidates the \emph{whole branch} of the extracted vertex (in the current spanning tree) instead of only the extracted vertex.

\begin{figure}[htbp]
\centering
\begin{tabular}{ccc}
\includegraphics[width=.3\textwidth]{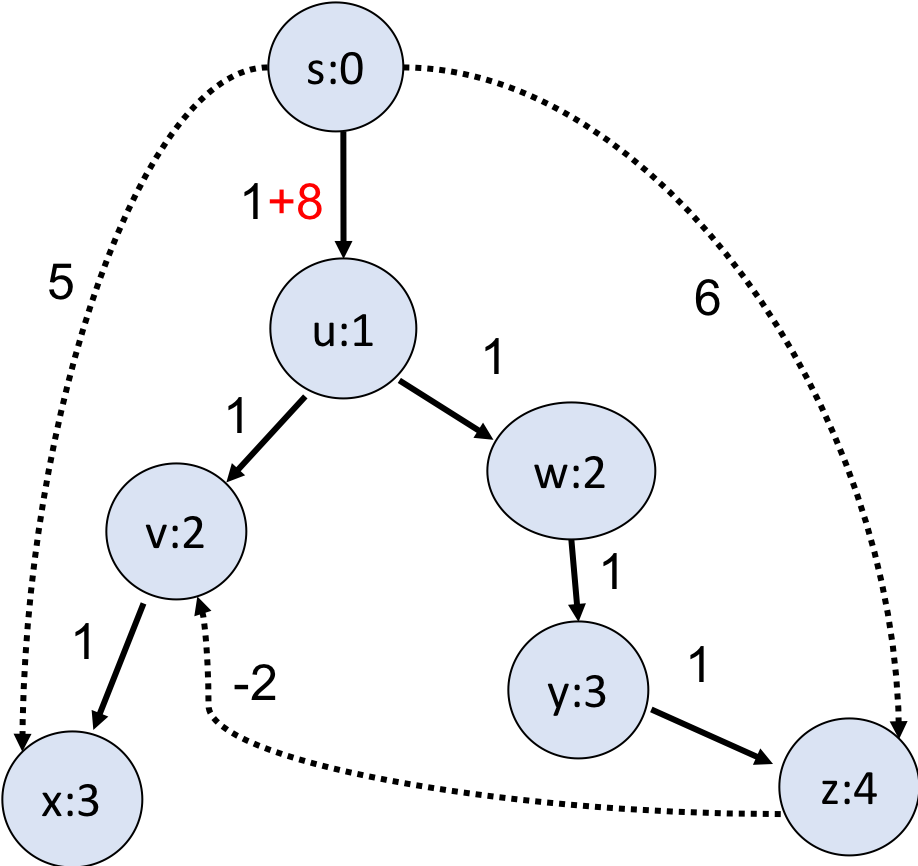}
&
\includegraphics[width=.3\textwidth]{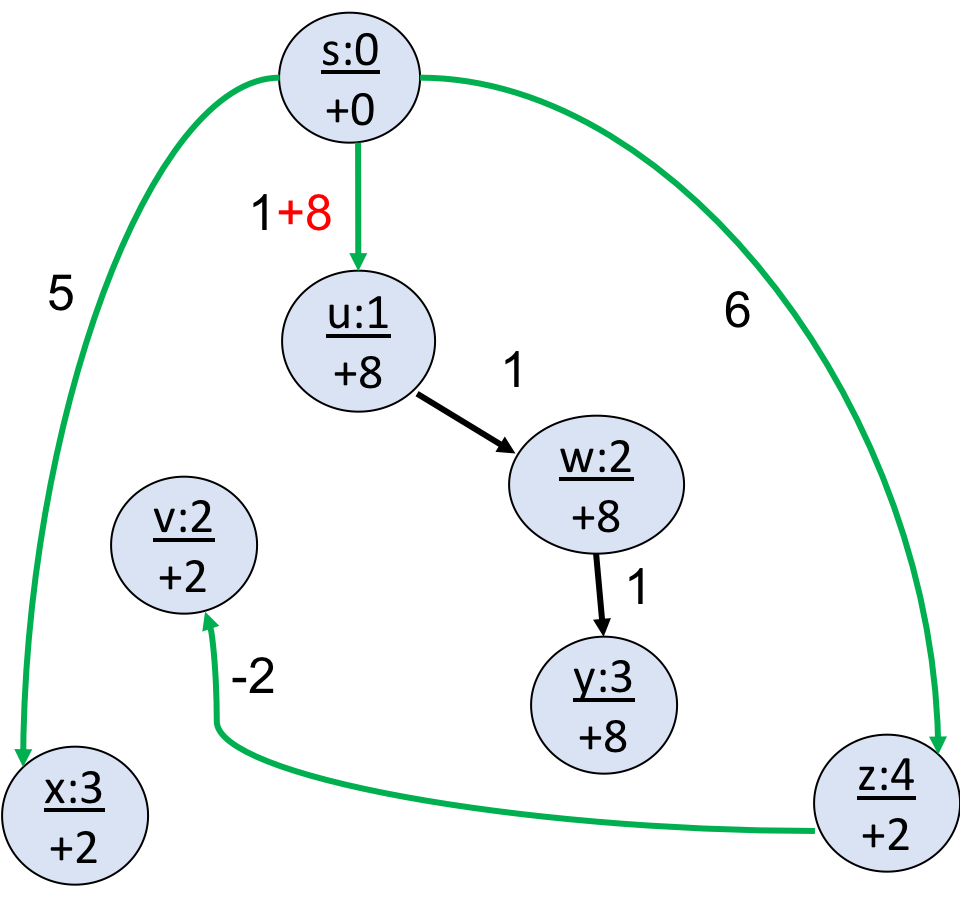}
& 
\includegraphics[width=.3\textwidth]{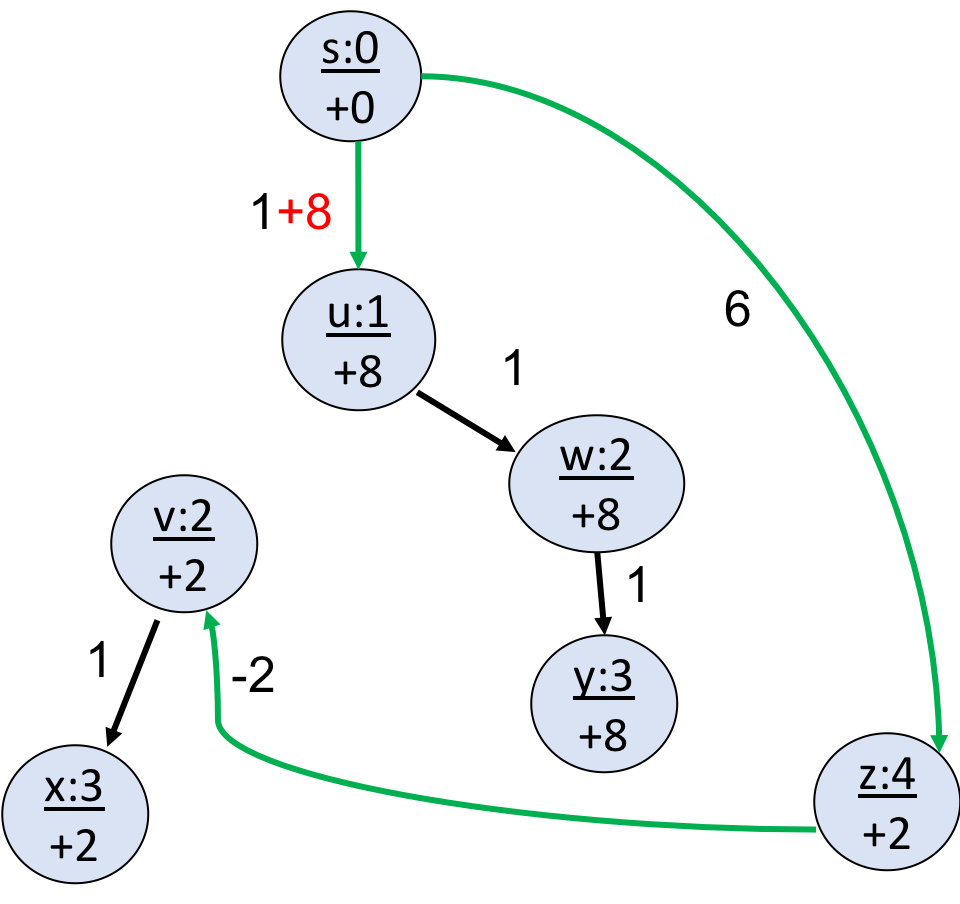}
\end{tabular}
\caption{A directed graph $G$ with weight $\omega(s,u)$ increased from 1 to 8 (left), the SPT of $G'$ obtained by Algorithm~\ref{alg:+} (centre), and the STP of $G'$ with minimal edge changes (right), where the original SPT of $G$ is illustrated in solid lines in the left, and we write, for example, $u:1$ and $\frac{u:1}{+8}$, to denote that the original distance of vertex $u$ is 1 and its updated distance is increased by 8.}  
\label{fig:graphA00}
\end{figure}

\begin{example}\label{example1}
Figure~\ref{fig:graphA00} shows an example, where $V=\{s,u,v,w,x,y,z\}$ and the weight of edge $e_0=(s,u)$ is increased from 1 to 9. At first, we have $X_0=\{s\}$ (the initial set of consolidated vertices), $Y_0=\ \downarrow\!u = \{u,v,w,x,y,z\}$ (the initial set of unsettled vertices). Consider all edges from $X_0$ to $Y_0$ except $e_0$. We have $\Delta_{T}(s,x)=5-3=2$, $\Delta_{T}(s,z)=6-4=2$. Since both are smaller than 8, the increment of edge $e_0$,  we put $\la x, (s,2)\ra$ and $\la z, (s,2)\ra$ in $Q$. Since the current lengths of $x,z$ are 5 and 6, respectively,  in the first iteration, we extract $\la x, (s,2)\ra$ and move the $x$ branch (which only contains $x$ itself) directly under $s$. The resulted spanning tree is $T_1$ and we have $X_1=\{x\}$ (the set of vertices consolidated in the first iteration), $Y_1=\{u,v,w,y,z\}$ (the set of unsettled vertices after the first iteration). Since there are no edges from $X_1$ to $Y_1$, we add no new vertices to $Q$. Then, in the second iteration,  we extract $\la z, (s,2)\ra$, which is the only vertex in $Q$, and move the $z$ branch of $T_1$ (which only contains $z$ itself) directly under $s$. The resulted spanning tree is $T_2$ and we have $X_2=\{z\}$ and $Y_2=\{u,v,w,y\}$. Since $(z,v)$ is the only edge from $X_2$ (the set of vertices consolidated in the second iteration) to $Y_2$  (the set of unsettled vertices after the second iteration), we examine it and find $\Delta_{T_2}(z,v) = (6-2) - 2 = 2<8$. We put $\la v,(z,2)\ra$ in $Q$. In the third iteration, we extract $\la v,(z,2)\ra$, which is the only vertex in the current $Q$,  and move the $v$ branch of $T_2$ (which only contains $v$ itself) under $z$. We have $X_3=\{v\}$  (the set of vertices consolidated in the third iteration) and $Y_3=\{u,w,y\}$ (the set of unsettled vertices after the second iteration). The resulted spanning tree is $T_3$ and we have an empty queue $Q$. Thus, the algorithm returns $T_3$ as output, which is shown in Figure~\ref{fig:graphA00} (centre). It is easy to see that $T_3$ is an SPT of $G'$. This SPT, however, is not the optimal one. In the right of Figure~\ref{fig:graphA00}, we give an SPT of $G'$ with minimal edge changes.
\end{example}

\begin{algorithm}
   \DontPrintSemicolon%
   \SetKwInOut{Input}{Input}%
   \SetKwInOut{Output}{Output}%
   \SetKw{KwSt}{s.t.}%
     \Input{%
       A shortest-path tree $T$ of $G=(V,E)$ and an edge $e_0=(x_0,y_0)$ with a larger weight $\omega'(x_0,y_0) > \omega(x_0,y_0)$%
     }%
     \Output{%
      A shortest-path tree of $G'=(V,E')$, where $E'$ differs from $E$ only in that the weight of $e_0$ is increased from $\omega(x_0,y_0)$ to $\omega'(x_0,y_0)$.%
     }%
     \BlankLine%
     \lIf{$e_0\not\in T$}{\Return $T$ \label{line:1+}}\Else{  
    $\Delta(y_0) \gets \omega'(x_0,y_0) - \omega(x_0,y_0)$ \label{line:3+} \;
 $Y\gets\ \downarrow\! y_0$  \tcp*{$Y$ is the set of unsettled vertices}
 $X \gets V\setminus Y$  \tcp*{$X$ is the set of newly settled vertices} 
 $Q \gets \varnothing$  \label{line:6+} 
\tcp*{Each element in $Q$ is of the form $\la n, (p,\Delta)\ra$, where $p$ denotes a potential parent for vertex $n$, 
and $\Delta$ denotes the potential distance change for $n$.}
\ForEach{\rm  $(x,y) \in E$ with $x \in X$ and $y\in Y$ \label{line:7+}} {%
\textsf{newdist} $\gets \distance_G(x) + \omega(x,y)$\;
		$\Delta \gets \mbox{\textsf{newdist}} - \distance_G(y)$\; 
	\lIf{\rm $\Delta < \Delta(y_0)$}{\textsf{ENQUEUE}$(Q, \la y,(x,\Delta)\ra)$ \label{line:10+}}
	\tcp*{The instruction adds one more element $y$ to $Q$. If vertex $y$ is already in $Q$, then the new attributes will replace the old ones only if the new $\Delta $ is smaller.}
}
\While{$Q\not=\varnothing$ \label{line:11+}}{
      $\la y_Q, (x_Q,\Delta) \ra \gets \textsf{EXTRACTMIN}(Q)$\; 
	\tcp*{When there are various vertices with the same minimum $\Delta$, we select any one with the smallest distance from $s$.} 
         $Pa(y_Q) \gets x_Q$ \label{line:13+}
         \tcp*{Change the parent of $y_Q$ in $T$ to $x_Q$}
                $X\gets\  \downarrow\! y_Q$       \label{line:14+} \quad\quad        \tcp*{$\downarrow\! y_Q$ is the set of descendants of $y_Q$ (including $y_Q$ itself) in the current tree} 
        $Y \gets Y \setminus X$        \tcp*{$Y$ is the set of unsettled vertices}
        \ForEach{$x\in X$}{
         $D(x) \gets \distance_G(x) + \Delta$  \label{line:17+}\;
       \lIf{$x \in Q$}{\rm \textsf{REMOVE}$(x,Q)$ \label{line:18+}} 
        }
	 \ForEach{\rm $(x,y)\in E$ with $x\in X$ and $y \in Y$ \label{line:19+}}{
	 \textsf{newdist} $\gets D(x) + \omega(x,y)  $\;
	 \If{ \rm \textsf{newdist} $< \distance_G(y)+\Delta(y_0)$} {
	 $\Delta \gets  \textsf{newdist} - \distance_G(y)$\;
	 \textsf{ENQUEUE}$(Q, \la y, (x,\Delta) \ra)$ \label{line:24+}}
	}
	 
	 }
\Return $T$
}
     \caption{The incremental algorithm} \label{alg:+}
 \end{algorithm}

\subsection{Correctness of the algorithm} 
Given input $G=(V,E)$ and edge $e_0=(x_0,y_0)$ with $\omega'(e_0)>\omega(e_0)$ as above, suppose $e_1=(x_1,y_1)$, ..., $e_k=(x_{k},y_{k})$ are the sequence of edges extracted by the algorithm before it terminates. We define
\begin{align}\label{eq:incXs}
X_0=V\setminus\!\downarrow\! y_0,\  X_1=\ \downarrow\! y_1 \setminus X_0,  ...,\ X_{k}=\ \downarrow\! y_{k} \setminus \bigcup_{i=0}^{k-1} X_{i},\ X_{k+1}=V\setminus \bigcup_{i=0}^k X_{i}.
\end{align}
Actually,  $X_1=\;\downarrow\! y_1$ as $y_1$ is taken from $\downarrow\! y_0$. By definition, these $X_i$ are pairwise disjoint and form a partition of the vertices of $G$.  For each $x \in V$, we define the incremental value of $x$ as the difference between the distances from $s$ to $x$ in $G'$ and $G$.
\begin{align} 
\label{eq:Delta(x)}
\delta(x) \equiv \distance_{G'}(x) - \distance_G(x).
\end{align} 

The following proposition asserts that vertices in $X_i$ ($1\leq i\leq k$) are consolidated in the $i$-th iteration \emph{collectively} and their distances from $s$ have a common increment value.

\begin{proposition}
\label{prop:Vi}
Suppose $T$ is an SPT of graph $G=(V,E)$ at root $s\in V$ and $e_0=(x_0,y_0)$ an edge in $T$ whose weight has been increased from $\omega(e_0)$ to $\omega'(e_0)$. Write $\theta=\omega'(e_0) - \omega(e_0)$. Let $e_i=(x_i,y_i)$ $(1\leq i\leq k)$ be the edge that is extracted by Algorithm~\ref{alg:+} in the $i$-th iteration, $X_i$ $(0\leq i\leq k+1)$ be defined as in \eqref{eq:incXs}, and $T_i$  the spanning tree constructed in the $i$-th iteration for $1\leq i\leq k$. 

Then, for each $1\leq i\leq k$, $X_i$ is the set of vertices consolidated in the $i$-th iteration and $x_i$ and $y_i$ are consolidated, respectively, before and in the $i$-th iteration, and $T_i$ is obtained by changing the parent of $y_i$ in $T_{i-1}$ to $x_i$, where $T_0=T$.
For any $u \in V$ and any $0\leq i\leq k+1$,  if $u \in X_{i}$ then $\distance_{G'}(u) = d_{T_{i}}(u)$ and $\delta(u) = \delta_{i}$, where $T_{k+1}=T_k$ and 
\begin{align}\label{eq:deltas+}
\hspace*{-3mm} \delta_0=0, \delta_{k+1}=\theta,  \delta_{i} \equiv \Delta_{T_{i-1}}(e_i) =  d_{T_{i-1}}(x_i) + \omega(x_i,y_i) - \distance_G(y_i)<\theta\ \mbox{\rm for}\ 1 \leq i \leq k. 
\end{align}
Moreover, we have $0=\delta_0 \leq \delta_1 \leq ... \leq \delta_i \leq \delta_{i+1} \leq ... \leq \delta_{k} < \delta_{k+1} = \theta$.
\end{proposition}
\begin{proof}
For any $u\in X_0$, by Lemma~\ref{lem3}, the path $\pi_{s,u}$ from $s$ to $u$ in $T$ has length $\distance_G(u)$, which cannot be improved and thus remains a shortest path in $G'$. Thus $\distance_{G'}(u)=d_{T_0}(u)=\distance_G(u)$ for any $u\in X_0$. 

We prove the claim inductively on $1\leq i\leq k$.  Suppose the result holds for any $j \leq i$ for some $0\leq i \leq k-1$. We prove this also holds for $j=i+1$. In the $(i+1)$-th iteration, the edge $e_{i+1}=(x_{i+1},y_{i+1})$ is extracted from $Q$. By the algorithm, this is possible only when $y_{i+1}$ was put in $Q$ in previous iterations (as specified by either line~\ref{line:10+} or line~\ref{line:24+}) and was not removed from $Q$ (as specified in line~\ref{line:18+}) before the $(i+1)$-th iteration. This implies that $x_{i+1}$ was consolidated before the $(i+1)$-th iteration and $y_{i+1}$ remains unsettled when we start the $(i+1)$-th iteration. Lines~\ref{line:14+}-\ref{line:18+} then consolidate all descendants of $y_{i+1}$ in $T_{i}$ (include $y_{i+1}$ itself) and remove all its descendants  from $Q$. By induction hypothesis, each $X_p$ with $p<{i+1}$ is the set of vertices consolidated in the $p$-th iteration. Thus the set of descendants of $y_{i+1}$ in $T_{i}$ is $X_{i+1} =\ \downarrow\! y_{{i+1}} \setminus \bigcup_{p=0}^{i} X_{p}$, which consists of all vertices being consolidated in the $(i+1)$-th iteration.

Suppose $u\in X_{i+1} =\ \downarrow\! y_{i+1} \setminus \bigcup_{j=0}^{i} X_{j}$. We show $\distance_{G'}(u) = d_{T_{i+1}}(u)$ and $\delta(u) = \delta_{i+1} < \theta$. Since $x_{i+1} \in \bigcup_{j=0}^i X_j$, by induction hypothesis, we have
$\distance_{G'}(x_{i+1})= d_{T_{i}} (x_{i+1})$. Therefore,  by 
$$\delta_{i+1} = \Delta_{T_{i}}(e_{i+1}) = d_{T_{i}}(x_{i+1}) + \omega(x_{i+1},y_{i+1}) - \distance_G(y_{i+1}),$$  we have $\delta_{i+1}<\theta$ as otherwise $e_{i+1}$ shall not be put in $Q$. Moreover, we have
$$\distance_{G'}(y_{i+1}) \leq \distance_{G'}(x_{i+1}) + \omega(x_{i+1},y_{i+1}) = d_{T_{i}} (x_{i+1}) + \omega(x_{i+1},y_{i+1}) = \distance_G(y_{i+1}) + \delta_{i+1}.$$
As $u\in\;\downarrow\! y_{i+1}$ in $T$, $y_{i+1}$ is on $\pi_{s,u}$
and $\distance_G(u) = \distance_G(y_{i+1}) + \length_{G}(\pi_{s,u}[y_{i+1},u])$. Thus 
\begin{align*}
 \distance_{G'}(u) &\leq \distance_{G'}(y_{i+1}) + \length_{G}(\pi_{s,u}[y_{i+1},u])\\  
 &\leq \distance_G(y_{i+1}) + \delta_{i+1} + \length_{G}(\pi_{s,u}[y_{i+1},u]) \\
 & = \distance_G(u) + \delta_{i+1}.
 \end{align*}
 
On the other hand, suppose $\pi=\la u_0=s, u_1, ..., u_\ell=u \ra$ is a shortest path in $G'$ from $s$ to $u$. We show $\distance_{G'}(u) = \length_{G'}(\pi) \geq  \distance_G(u) + \delta_{i+1}$. Assume $u_{\ell^*}$ is the last vertex on $\pi$ such that $u_{\ell^*}$ is in $X_{j'+1}$ for some $j'<i$. If $\Delta_{T_{j'+1}}(u_{\ell^*}, u_{{\ell^*}+1})  \geq \theta$, then, by  $\Delta_{T_{i}}(x_{i+1},y_{i+1}) = \delta_{i+1} <\theta$, we have 
\begin{align}\label{eq:ell*}
 d_{T_{j'+1}}( u_{\ell^*}) + \omega(u_{\ell^*},u_{{\ell^*}+1}) - \distance_G(u_{{\ell^*}+1})  & =  \Delta_{T_{j'+1}}(u_{\ell^*}, u_{{\ell^*}+1})  \geq \theta > \delta_{i+1}.
\end{align}
 Suppose $\Delta_{T_{j'+1}}(u_{\ell^*}, u_{{\ell^*}+1})  < \theta$. Then in the end of the $({\ell^*}+1)$-th iteration, $u_{{\ell^*}+1}$ is put in $Q$ if it is not there yet (cf. line~\ref{line:24+}).
Since $u_{{\ell^*}+1} \not\in \bigcup_{j=0}^{i} X_j$, 
it can not be extracted or removed from $Q$ before $y_{i+1}$. 
Thus, immediately before $(x_{i+1},y_{i+1})$ is extracted from $Q$, $u_{\ell^*+1}$ is in $Q$ with some priority value $\Delta\leq \Delta_{T_{j'+1}}(u_{\ell^*}, u_{{\ell^*}+1})$. Because we select $y_{i+1}$ over $u_{\ell^*+1}$ in the $(i+1)$-th iteration, the priority value of $u_{\ell^*+1}$ is not smaller than that of $u_{i+1}$, viz., $\delta_{i+1}$. Consequently, Eq.~\eqref{eq:ell*} also holds in this case.

Because $j'<i$ and, by induction hypothesis, $\distance_{G'}(u_{\ell^*}) = d_{T_{j'+1}}( u_{\ell^*})$,  Eq.~\eqref{eq:ell*} implies
$$\distance_{G'}(u_{\ell^*}) + \omega(u_{\ell^*}, u_{{\ell^*}+1}) \geq \distance_G(u_{{\ell^*}+1}) + \delta_{i+1}.$$ 
Furthermore, we have 
\begin{align*}
\distance_{G'}(u) &= \length_{G'}(\pi)\\
 &= \distance_{G'}(u_{\ell^*}) + \omega(u_{\ell^*}, u_{{\ell^*}+1}) + \length_{G}(\pi[u_{{\ell^*}+1},u_\ell=u]) \\
& \geq \distance_G(u_{{\ell^*}+1})  + \length_{G}(\pi[u_{{\ell^*}+1},u_\ell=u])  + \delta_{i+1} \\
& \geq \distance_G(u) + \delta_{i+1}.
\end{align*}
Therefore, we have $\distance_{G'}(u)=\distance_G(u) + \delta_{i+1}$ whenever $u\in X_{i+1}$.

We then show $\distance_{G'}(u) = d_{T_k}(u)$ for any $u\in T_{k+1}$. As $e_k=(x_k,y_k)$ is the last edge we extract, after removing all descendants of $y_k$ in $T_{k+1}$ from $Q$, the queue $Q$ is empty and no new vertices will be put in $Q$. This implies, for any edge $e=(x,y)$ with $x \in \bigcup_{i=0}^k X_i$, $y\in X_{k+1} \equiv V\setminus\bigcup_{i=0}^k X_i$, we have $\Delta_{T_{k}}(e) \geq \theta$. Therefore, given $u\in X_{k+1}$ and any path $\pi=\la u_0=s, u_1, ..., u_\ell=u \ra$ from $s$ to $u$ in $G'$, assume that  $u_{\ell^*}$ is the last vertex such that $u_{\ell^*}$ is in $X_{j'+1}$ for some $j'<k$. Then $u_{\ell^*+1}$ is in $X_{k+1}$. Because $d_{T_{j'+1}}( u_{\ell^*}) + \omega(u_{\ell^*},u_{{\ell^*}+1}) - \distance_G(u_{{\ell^*}+1}) = \Delta_{T_{j'+1}}(u_{\ell^*}, u_{{\ell^*}+1})  \geq \theta$ and $\distance_{G'}(u_{\ell^*})=d_{T_{j'+1}}(u_{\ell^*})$, we have $\distance_{G'}(u_{\ell^*})+ \omega(u_{\ell^*},u_{{\ell^*}+1}) \geq \distance_G(u_{{\ell^*}+1}) + \theta$.
Thus  
$ \length_{G'}(\pi) \geq \distance_{G'}(u_{\ell^*}) + \omega(u_{\ell^*}, u_{{\ell^*}+1}) + \length_{G}(\pi[u_{{\ell^*}+1},u_\ell=u])   \geq \distance_G(u) + \theta$. Because $\pi$ is an arbitrary path from $s$ to $u$ in $G'$, we have $\distance_{G'}(u) \geq \distance_G(u) + \theta$. On the other hand, let $\pi_{s,u}$ be the path from $s$ to $u$ in $T$. Then $\length_{G'}(\pi_{s,u})=\distance_G(u) + \theta$. This proves  $ \distance_{G'}(u)=\distance_G(u) + \theta$ for any $u\in X_{k+1}$. In particular, $\pi_{s,u}$ remains a shortest path in $G'$. This shows that $\distance_{G'}(u) = d_{T_k}(u)$ for any $u\in T_{k+1}$.  

Lastly, we show $0=\delta_0 \leq \delta_1 \leq ... \leq \delta_i \leq \delta_{i+1} \leq ... \leq \delta_{k+1} = \theta$. That $\delta_0=0$, $\delta_{k+1}=\theta$, and $0\leq \delta_i \leq \theta$ is clear. We need only show that $\delta_i\leq \delta_{i+1}$ for any $1\leq i\leq k-1$.  Since  $\delta_i=\Delta_{T_{i-1}}(e_i)$ is the smallest value when $e_i$ is extracted from $Q$, there is a descendant $u$ of $y_i$ in $T_{i}$ such that $e_{i+1}=(u,y_{i+1})$. Now, since $\distance_G(u) + \omega(u,y_{i+1}) \geq \distance_G(y_{i+1})$, $\distance_{G'}(u) = \distance_G(u) + \delta_i$, and $\delta_{i+1} = \Delta_{T_i}(u,y_{i+1})$, we have
\begin{align*}
\delta_{i+1}\!  =\! \distance_{G'}(u)\! +\! \omega(u,y_{i+1})\! -\! \distance_G(y_{i+1})\! =\! \distance_G(u)\! +\! \delta_i\! +\! \omega(u,y_{i+1})\! -\! \distance_G(y_{i+1})  \!\geq\! \delta_i .
\end{align*} 
This ends the proof.
\end{proof}

The above result asserts that the algorithm outputs an SPT of $G'$ for any input.
\begin{theorem} \label{thm:sound+}
Algorithm~\ref{alg:+} is sound and complete. Precisely, suppose $T$ is an SPT of a directed graph $G$ at root $s$ and $e_0=(x_0,y_0)$ an edge in $G$ whose weight has been increased from $\omega(e_0)$ to $\omega'(e_0)$.  The algorithm stops in at most $n-1$ iterations and outputs a spanning tree $T'$ that is an SPT of $G'$, where $n$ is the number of vertices in $G$. 
\end{theorem}
\begin{proof}
If $e_0$ is not in $T$, then, by Lemma~\ref{lem3}, $T$ remains an SPT of $G'$ and the algorithm correctly outputs $T$ as specified in line~\ref{line:1+}. The case when $e_0$ is in $T$ follows directly from Proposition~\ref{prop:Vi}. We only note here that,  in the $(i+1)$-th iteration ($0\leq i \leq  k-1$), the spanning tree $T_{i+1}$ differs from the previous spanning tree $T_i$ constructed in the $i$-th iteration only in that, in $T_{i+1}$, the subtree with root $y_{i+1}$ in $T_i$  is moved under $x_{i+1}$. For any $0\leq j\leq k+1$ and any vertex $u\in X_j$, the path $\pi$ from $s$ to $v$ in $T_{j}$ is fixed from then on, i.e., $\pi$ remains a path in $T_{k+1}=T_k$ (cf. Lemma~\ref{lem:2trees}). This implies that $\distance_{G'}(u)=d_{T_j}(u)= d_{T_{k}}(u)$ for any $u \in X_j$ and any $0\leq j\leq k+1$. Because $X_0,X_1,...,X_{k+1}$ is a partition of the vertex set of $G$, we know $T_{k}$ (the output of Algorithm~\ref{alg:+}) is an SPT of $G'$.
\end{proof}

We next analyse the computational complexity of the algorithm. Here, similar to \cite{RamaR96}, we measure the computational complexity of a (semi-)dynamic algorithm in terms of the output changes. We say a vertex $v$ is \emph{affected} if either its distances from $s$ or its parents in $T$ and $T'$ are different, and say $v$ is \emph{strongly affected} if its parents in $T$ and $T'$ are different. It is easy to see that a vertex is  strongly affected if and only if it is extracted by the algorithm. Suppose the weight of $(x_0,y_0)$ is increased and $(x_0,y_0)$ is an edge in $T$. Then the set of affected vertices is $\downarrow\! y_0$, which is highly depended on the input SPT $T$. Moreover, if there are $k$ extractions, then the set of affected vertices is  $\downarrow\! y_0 =\bigcup_{i=1}^{k+1} X_i$. 

\begin{theorem} \label{thm:comp+}
The time complexity of the algorithm is $O(m_0+n_0 \log n_0)$ if $Q$ is implemented as a relaxed heap, where  $n_0$ is the number of affected vertices, $m_0$ is the number of edges in $G$ whose heads are affected vertices.
\end{theorem}
\begin{proof}
 In the algorithm, we only examine edges which have tails in $\downarrow\!y_0$. Let $(x,y) $ be such an edge. Then $(x,y)$ is examined in the $i$-th iteration if and only if $x\in X_{i}$ and $y\in Y_{i}$, where $X_i$ is defined as in Eq.~\eqref{eq:incXs}, $Y_0 =\ \downarrow\! y_0$ and $Y_{i} = Y_{i-1} \setminus X_i$ for $1\leq i\leq k$. Because $X_i$'s are pairwise disjoint, each $(x,y)$ is examined at most once (lines 7 and 19). So there are at most $m_0$ edge examinations.

Let $n_s$ be the number of extractions (i.e., the number of strongly affected vertices). The algorithm extracts $n_s$ vertices from $Q$ (i.e., $n_s$ EXTRACTMIN operations) and thus runs in $n_s$ iterations. The maximal number of insertions into $Q$ (i.e., the maximal size of $Q$) is $n_0$ and the number of  ENQUEUE operations (i.e., the number of vertex insertions and key decrements in $Q$) is at most $m_0$. There are at most $n_0-n_s$ edge deletions (i.e., REMOVE operations). 

Since $Q$ is implemented as a relaxed heap \cite{DriscollGST88}, each insertion/key modification runs in constant time while each deletion/remove runs in $O(\log n_0)$ time. Thus the time complexity of the algorithm is $O(m_0+n_0 \log n_0)$.
\end{proof}

One potential problem with the above analysis is that we don't know if the number of edge changes to $T$ is minimal. In the worst case, $n_s$ could be as bad as $n_0$. In the following subsection, we consider how to obtain an SPT with minimal edge changes.

\subsection{Minimal edge changes}
In \cite{NarvaezST01}, it is proved that the SPT $T'$ constructed by Algorithm~\ref{alg:+} has the minimal edge changes to $T$ if $G$ has only positive weights. This is, however, not true when negative weights present. One example is shown in Figure~\ref{fig:graphA00}, where the SPT $T'$ constructed by Algorithm~\ref{alg:+} is shown in the centre. Note that there are four edge changes in $T'$ (we regard $(s,u)$ as a changed edge). However, if we move $x$ directly under $v$, then we obtain another SPT with three edge changes. In this subsection, we show how to obtain an SPT with minimal edge changes from the SPT $T'$ output by Algorithm~\ref{alg:+}.

We introduce the following notion.

\begin{definition}\label{dfn:branches}
Given input $G$, $T$, $e_0=(x_0,y_0)$, $\omega(e_0)<\omega'(e_0)$ of Algorithm~\ref{alg:+}, suppose $e_0$ is an edge in $T$ and $\widehat{T}$ is an SPT of $G'$.  We define a \emph{$\widehat{T}$-branch} as a connected component of the directed graph obtained by removing the edge $e_0$ and  all edges that are not in $T$ from $\widehat{T}$. We call the $\widehat{T}$-branch that contains $s$ the root $\widehat{T}$-branch and call the root of a $\widehat{T}$-branch a \emph{miniroot}. We say two $\widehat{T}$-branches $\mathcal{B}_1$ and $\mathcal{B}_2$ are \emph{linked} if there exists an edge $(u_1,u_2)$ in $T$ such that   $u_1,u_2$ belong to either of, but not both,  $\mathcal{B}_1$ and $\mathcal{B}_2$. 
\end{definition}
Note that if there are $k$ new edges in $\widehat{T}$, then, after they (as well as edge $e_0$ if it is in $\widehat{T}$) are removed from $\widehat{T}$, we have all together $k+1$ (or $k+2$ if $e_0$ is in $\widehat{T}$) $\widehat{T}$-branches. As a connected component, each $\widehat{T}$-branch is a subtree of both $\widehat{T}$ and $T$. This implies in particular that, for any $\widehat{T}$-branch $\mathcal{B}$ and any two vertices $u_1,u_2$ in $\mathcal{B}$, if $u_1$ is an ancestor of $u_2$ in $T$, then $\mathcal{B}$ contains the whole path from $u_1$ to $u_2$ in $T$. Consider the example shown in Figure~\ref{fig:graphA00} again. Let $T'$ be the SPT of $G'$ output by Algorithm~\ref{alg:+} (shown in the centre). It has five $T'$branches, viz. the root branch $\mathcal{B}_0=\{s\}$, $\mathcal{B}_1=\{x\}$, $\mathcal{B}_2=\{z\}$, $\mathcal{B}_3=\{v\}$ and $\mathcal{B}_4=\{u,w,y\}$.  Because $(v,x)$ is an edge in $T$, the two branches $\mathcal{B}_1$ and $\mathcal{B}_3$ are linked. Similarly,  $\mathcal{B}_4$ is linked to all of $\mathcal{B}_0$, $\mathcal{B}_3$, and $\mathcal{B}_2$.

\begin{figure}[htbp]
\centering
\begin{tabular}{ccc}
\includegraphics[width=.28\textwidth]{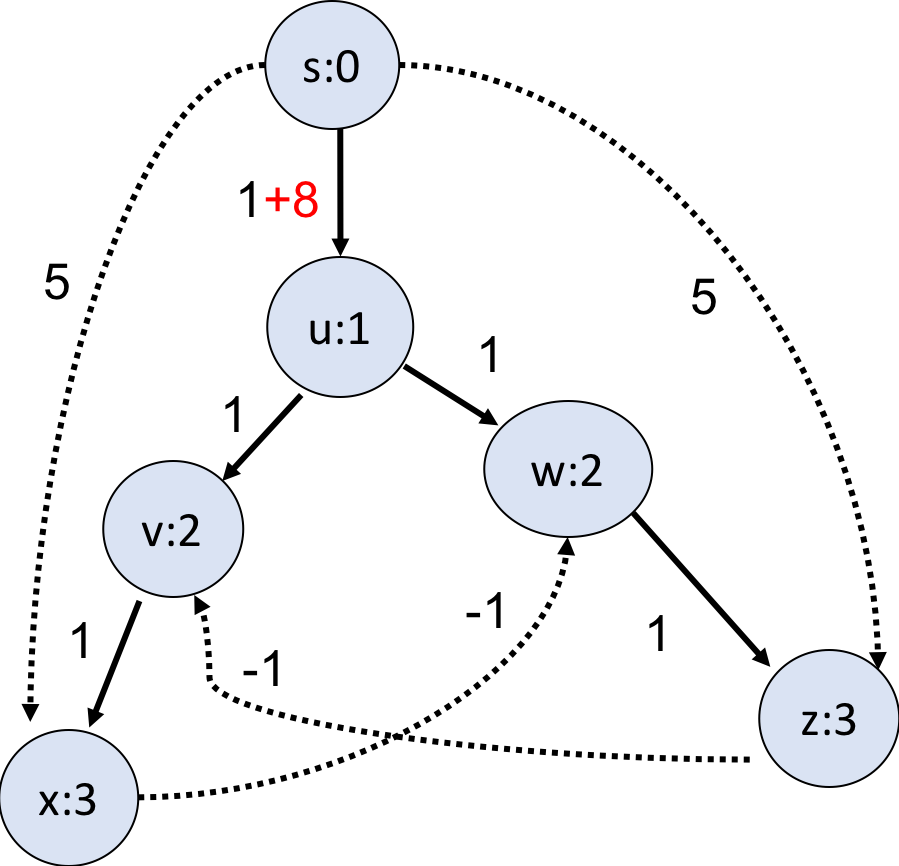}
&
\includegraphics[width=.28\textwidth]{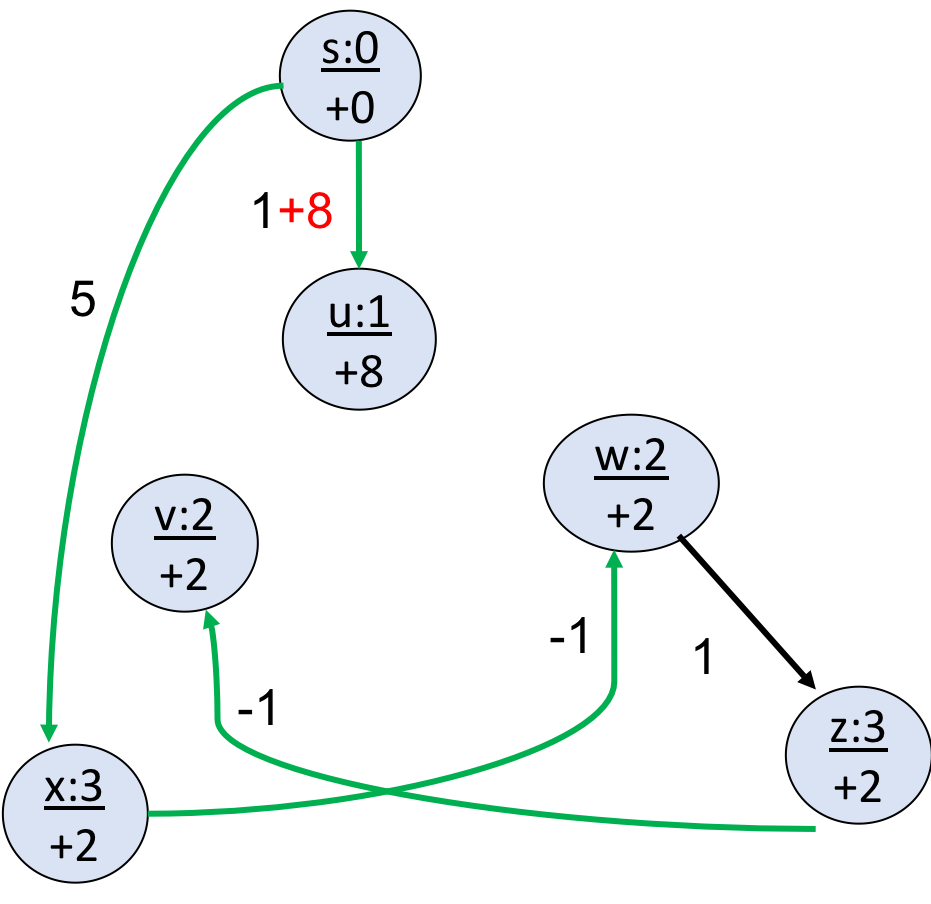}
& 
\includegraphics[width=.28\textwidth]{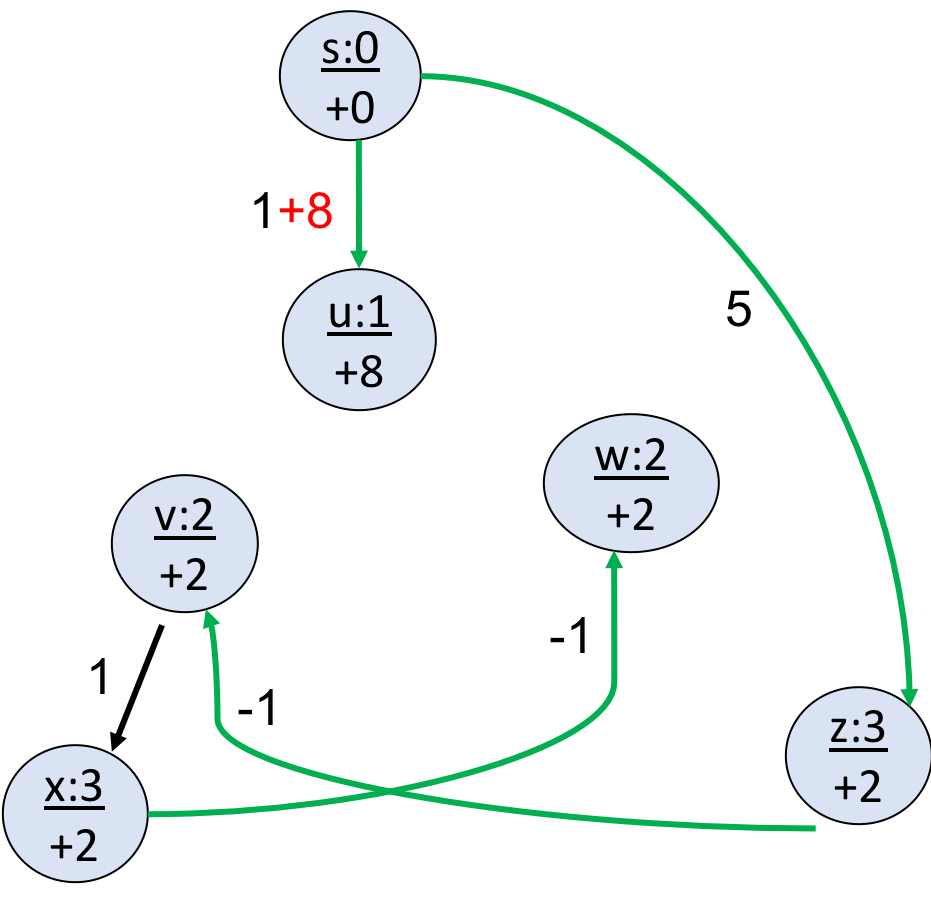}
\end{tabular}
\caption{A directed graph with 0-cycles and its SPT $T$ (left), its two SPTs $T'_1$ (centre) and $T'_2$ after edge weight increase, where we write, for example, $\frac{w:2}{+2}$, to denote that the original distance of vertex $w$ is 2 and its updated distance is increased by 2.}  \label{fig:0cycle_aG}
\end{figure}
\begin{remark}\label{remark}
In \cite{NarvaezST01}, Narvaez et al. introduced a related notion of `branch'. Our notions of `miniroot' and `root branch' are borrowed from there. Given an old SPT $T$ and an edge weight increase or decrease, they  define that two vertices are in the same branch `if they are are connected in both the old SPT and in \emph{any} new SPT by some edge that does not change weight'. If the graphs contain only edges with positive weights, then a branch in their sense corresponds to a $\widehat{T}$-branch for some \emph{optimal} SPT $\widehat{T}$ of $G'$, i.e., SPT with minimal edge changes to $T$. This, however, is not the case when the graph has 0-cycles. Consider the directed graph $G$ in Figure~\ref{fig:0cycle_aG} (left) as an example, where we have a 0-cycle $\la x,w,z,v,x\ra$.\footnote{A similar example of directed graph with non-negative weights is obtained by changing the weights of the four edges in the 0-cycle as 0.} According to \cite[Definition 1]{NarvaezST01}, $v,x$ are in the same branch (since $(v,x)$ is an edge in $T'_2$ (right)), so are $w,z$  (since $(w,z)$ is an edge in $T'_1$ (centre)). Thus there are four branches in this graph. It is not difficult to see that there is no SPT $\widehat{T}$ of $G'$ with minimal edge changes such that both $\{w,z\}$ and $\{v,x\}$ are $\widehat{T}$-branches.

Another problem with their definition is that, when negative weights present, these branches cannot always be obtained from their {\sf Ball-String} algorithm. Consider the graph shown in Example~\ref{example1} for such an example, where the SPT of $G'$ output by {\sf Ball-String} is shown in the centre, which has one more edge change than the SPT shown in the right.

Their proof of the correctness and minimality of the output SPT (i.e., \cite[Theorem~1]{NarvaezST01}) heavily depends on the assumption that these branches correspond to the $\widehat{T}$-branches of some optimal $\widehat{T}$. The above analysis shows that their proof cannot be extended to the case when directed graphs have arbitrary weights. In addition, their assertion in \cite[Lemma~3]{NarvaezST01} that ``After the first iteration, branches $\mathcal{K}_0$ and $\mathcal{K}_1$ are consolidated", which was used as the basis step of the proof of \cite[Theorem~1]{NarvaezST01}, is not true even for directed graphs with only positive weights. Indeed, in the decremental case and when $\distance_G(x_0)+\omega'(e_0) < \distance_G(y_0)$, the root branch is determined only in the last iteration, after all possible improvement have been made for all affected vertices.

\end{remark}


We now come back to our definition of $\widehat{T}$-branches. For any two vertices $x,y$ in $\widehat{T}$-branch $\mathcal{B}$, we show they have the same $\delta$ value.
\begin{lemma}\label{lem:constant}
Suppose $\widehat{T}$ is an SPT of $G'$ and $\mathcal{B}$ is a $\widehat{T}$-branch with miniroot $u_0 \not =s$. For any $x\in \mathcal{B}$, we have $\delta(x)\equiv \distance_{G'}(x) - \distance_G(x) = \distance_{G'}(u_0) - \distance_G(u_0) \equiv \delta(u_0)$.
\end{lemma}
\begin{proof} 
By the definition of $\widehat{T}$-branch, $u_0$ is on $\pi_{s,x}$, the path from $s$ to $x$ in $T$, and $\pi_{s,x}[u_0,x]$ is contained in $\mathcal{B}$ and thus also the tree path from $u_0$ to $x$ in $\widehat{T}$.  Because 
\begin{align*}
\distance_{G'}(x) &=\distance_{G'}(u_0) + \length_{G'}(\pi_{s,x}[u_0,x])\\ 
&= \distance_{G}(u_0) + \delta(u_0) + \length_{G}(\pi_{s,x}[u_0,x])\\
& = \distance_{G}(x) + \delta(u_0),
\end{align*}
we have $\delta(x)=\delta(u_0)$.
\end{proof}
By the above result, we also call $\delta(u_0)$ the $\delta$ value of  the $\widehat{T}$-branch $\mathcal{B}$.
\begin{proposition}
\label{lem:cbranches}
Let $\mathcal{B}_1$ and $\mathcal{B}_2$ be two linked $\widehat{T}$-branches with miniroots $x_1$ and $x_2$ respectively. 
Assume $\delta(x_1)=\delta(x_2)=\delta$ and $0<\delta <\theta$. Then $x_1$ is an ancestor of $x_2$ in $T$ or vice versa. Suppose $x_1$ is an ancestor of $x_2$ in $T$, and, in addition, assume that $x_2$ is an ancestor of $x_1$ in $\widehat{T}$. Then $G$ has a 0-cycle. 
\end{proposition}
\begin{proof}
Without loss of generality, suppose $y_1\in \mathcal{B}_1$, $y_2 \in \mathcal{B}_2$ and $(y_1,y_2)$ is an edge in $T$. By definition of $\widehat{T}$-branches, $\mathcal{B}_1$ and $\mathcal{B}_2$ are disjoint subtrees of both $T$ and $\widehat{T}$. Since $x_1$ is an ancestor of both $y_1$ and $y_2$ in $T$ and $y_1$ is not in $\mathcal{B}_2$, we must have $x_2=y_2$. Thus $x_1$ is an ancestor of $x_2$ in $T$. Suppose in addition that $x_2$ is an ancestor of $x_1$ in $\widehat{T}$. Let $\pi_1=\pi_{s,x_2}[x_1,x_2]$ be the path from $x_1$ to $x_2$ in $T$, and $\pi_2$ the path from $x_2$ to $x_1$ in $\widehat{T}$. We have $\distance_{G}(x_2) = \distance_{G}(x_1) + \length_{G}(\pi_1)$ and $\distance_{G'}(x_1) = \distance_{G'}(x_2) + \length_{G'}(\pi_2)$. Since $\delta=\delta(x_1)=\delta(x_2)$, we have $\distance_{G'}(x_2)  = \distance_G(x_2) + \delta$ and $\distance_{G'}(x_1) = \distance_G(x_1) +\delta$. Moreover, since $0<\delta<\theta$ and  by Proposition~\ref{prop:Vi}, both $x_1$ and $x_2$ are in $\downarrow\! y_0$. Note that $e_0=(x_0,y_0)$ is not in $\pi_1$ because $x_1$ is an ancestor of $x_2$ in $T$ and $x_1\! \in\; \downarrow\! y_0$. Meanwhile, $e_0$ is not in $\pi_2$. This is because, otherwise, we shall have $\distance_{G'}(x_0)=\distance_{G'}(x_2)+\length_{G'}(\pi_2[x_2,x_0])=\distance_G(x_2)+\delta+\length_{G}(\pi_2[x_2,x_0])\geq \distance_G(x_0)+\delta>\distance_G(x_0)$, also a contradiction with  $\delta(x_0)=0$. By Lemma~\ref{lem:len_paths}, this implies $\length_{G'}(\pi_2)=\length_G(\pi_2)$. Combining the above equations, we have 
\begin{align*}
\distance_G(x_1) =\distance_{G'}(x_1) - \delta &= \distance_{G'}(x_2) + \length_{G'}(\pi_2) -\delta\\
& = \distance_{G}(x_2) + \length_{G'}(\pi_2) \\
&=  \distance_G(x_1) + \length_G(\pi_1)+ \length_G(\pi_2),
\end{align*} thus $\pi_1+\pi_2$ is a 0-cycle in $G$.
\end{proof}

It seems not too strict to require that the directed graph $G$ have no 0-cycles. If this is the case, we may merge two linked branches with the same $\delta$ value to get a better SPT for $G'$. Indeed, suppose $T'$ is an SPT of $G'$ and $\mathcal{B}_1$ and $\mathcal{B}_2$ are two linked $T'$-branches such that $u_i\in\mathcal{B}_i$ for $i=1,2$ and $(u_1,u_2)$ is an edge in $T$. As in the proof of Proposition~\ref{lem:cbranches}, we know $u_2$ is the miniroot of $\mathcal{B}_2$. By merging $\mathcal{B}_1$ and $\mathcal{B}_2$, we mean that we change the parent of $u_2$ in $T'$ back to $u_1$. In this way, we get an SPT of $G'$ with fewer edge changes than $T'$. Consider the directed graph $G$ in Figure~\ref{fig:graphA00} (left) again, where $\mathcal{B}_1=\{x\}$ is linked to $\mathcal{B}_3=\{v\}$ and $\delta(x)=\delta(v)=2$. After merging them, we get the branches of the optimal STP $T'_2$ (shown in the right of Figure~\ref{fig:graphA00}) of $G'$. In case $G$ has 0-cycles, this operation is sometimes not correct. See Figure~\ref{fig:0cycle_aG} (centre) for such an example, where $T'$-branch $\{v\}$ is linked with $T'$-branch $\{x\}$ by edge $(v,x)$ in $T$ and $\delta(v)=\delta(x)=2$, but we cannot merge them to get an even better STP for $G'$.

\begin{theorem} \label{thm:minchange}
Suppose $G$ has no 0-cycles. Let $T$ be the SPT of $G$ in the input of Algorithm~\ref{alg:+} and $T'$ any SPT of $G'$. Construct $T^*$ by merging any two linked $T'$-branches with the same $\delta$ value. Then $T^*$ is an SPT of $G'$ and has minimal edge changes to $T$.
\end{theorem}  
\begin{proof}
Without loss of generality, we assume that $T'$ has the same root branch as $T$ and at most one branch with $\delta$ value $\theta$.
We first show that $T^*$ remains an SPT of $G'$. Let $\mathcal{B}_1$ and $\mathcal{B}_2$ be two linked $T'$-branches with $u_1\in \mathcal{B}_1$, $u_2\in \mathcal{B}_2$, $\delta(u_1)=\delta(u_2)=\delta$, and $(u_1,u_2)\in T$. Then $u_2$ is the miniroot of $\mathcal{B}_2$. For any $u\in \mathcal{B}_2$, let $\pi_1$ be the path from $s$ to $u_1$ in $T'$ and let $\pi_2$ be the path from $u_2$ to $u$ in $T$. Since $u_1$ is the parent of $u_2$ and $u_2$ an ancestor of $u$ in $T$, we have $\distance_G(u)=\distance_G(u_1) + \omega(u_1,u_2)+ \length_{G}(\pi_2)$. By $e_0\not \in \pi_2$ and Lemma~\ref{lem:len_paths}, we also have $\length_{G}(\pi_2)=\length_{G'}(\pi_2)$. Moreover, by Lemma~\ref{lem:constant}, we have $\delta(u)=\distance_{G'}(u) - \distance_G(u) =\delta$. Therefore, we have $\distance_{G'}(u)=\distance_G(u)+\delta = \distance_G(u_1) + \omega(u_1,u_2)+ \length_{G}(\pi_2) + \delta = \distance_{G'}(u_1) + \omega(u_1,u_2) +  \length_{G}(\pi_2) = \length_{G'}(\pi_1+(u_1,u_2)+\pi_2)$. This implies that, after merging $\mathcal{B}_1$ and $\mathcal{B}_2$, the resulted spanning tree remains an SPT of $G'$. Continuing this process until there are no linked branches with the same $\delta$ value, the resulted spanning tree $T^*$ remains an SPT of $G'$.

We next prove that $T^*$ has minimal edge changes. Let $T^\dag$ be an SPT of $G'$ which has the minimal edge changes to $T$. Define a $T^\dag$-branch as in Definition~\ref{dfn:branches}. We prove that each $T^\dag$-branch is a $T^*$-branch and vice versa. 
Suppose this is not the case and there is a $T^\dag$-branch $\mathcal{B}$ that overlaps with at least two $T^*$-branches. Clearly, there exists an edge $(u_1,u_2)$ in $\mathcal{B}$ such that they belong to different $T^*$-branches. Since $u_1,u_2$ are in the same $T^\dag$-branch, by Lemma~\ref{lem:constant}, we have $\delta(u_1)=\delta(u_2)$. Therefore, the two linked $T^*$-branches should have been merged, a contradiction. Therefore,  $\mathcal{B}$ can overlap with only one $T^*$-branch. That is, each $T^\dag$-branch is contained in a unique $T^*$-branch. On the other hand, because $T^\dag$ has the minimal edge changes to $T$, it has no more branches than $T^*$ does. Consequently, the $T^\dag$-branches are the same as the $T^*$-branches and, thus, $T^*$ is an SPT which has the minimal edge changes to $T$.  
\end{proof}

The procedure described above can be achieved by adding several lines of pseudocode to Algorithm~\ref{alg:+}, see Algorithm~\ref{alg:+r}. We first introduce (lines 1-2) a value $\lambda$ to denote the $\delta$ value of a branch and a set $\Sigma$ to collect all branches which have the same $\delta$ value $\lambda$. In each iteration, whenever a vertex $\la y_Q, (x_Q, \Delta)\ra$ is extracted, we do not update the parent of $y_Q$ immediately. Instead, we record its original parent as $Pt(y_Q)$ and then check if $\Delta>\lambda$. If so, then, by $\delta_i \leq \delta_{i+1}$ for $0\leq i \leq k$ (see Proposition~\ref{prop:Vi}),  $y_Q$ leads to a new branch that has bigger $\delta$ value and the set $\Sigma$ is full now, i.e., it contains the miniroots of all branches which have $\delta$ value $\lambda$. For each $y\in \Sigma$, we then check if $Pt(y)$, the original parent of $y$, has the same $\delta$ value as $y$. To this end, we need only check if  $\distance_{G'}(Pt(y))=d_{\widehat{T}}(Pt(y)) $ is identical to $\distance_G(Pt(y)) + \lambda$, where $\widehat{T}$ is the current spanning tree of $G'$. If this is the case, then we restore the parent of $y$ and the two linked branches containing $y$ and, respectively, $Pa(y)$ are merged. After all vertices in $\Sigma$ have been checked and all linked branches are merged, we empty $\Sigma$. No matter if $\Delta>\lambda$ or not, we next put $\la y_Q, Pt(y_Q) \ra$ in $\Sigma$ and collect all miniroots which have the $\delta$ value $\Delta$. 
Here we maintain $\Sigma$ as a list and each miniroot will be put in $\Sigma$ once and examined once. This shows that the revised while loop only introduces $O(n_s)$ extra operations, where $n_s<n_0$ is the number of extractions in Algorithm~\ref{alg:+}.

 \begin{algorithm}
   \DontPrintSemicolon%
 	$\lambda\gets\ 0$  \tcp*{$\lambda$ is the $\delta$ value of an iteration}
	$\Sigma \gets \varnothing$   \tcp*{$\Sigma$ is the set of extracted vertices with the same $\delta$ value}
	\While{$Q\not=\varnothing$}{
      		$\la y_Q, (x_Q,\Delta) \ra \gets \textsf{EXTRACTMIN}(Q)$  \tcp*{line 12 of Algorithm~\ref{alg:+}} 
	        $Pt(y_Q) \gets Pa(y_Q)$ \tcp*{$Pt(y_Q)$ stores the original parent of $y_Q$ in $T$}
                 \If{ $ \Delta > \lambda$} {
			\ForEach{$y\in \Sigma$}{
			        \If{$D(Pt(y)) = \distance_G(Pt(y)) + \lambda$ }
			        {$Pa(y) \gets Pt(y)$ \tcp*{The parent of $y$ is restored and two branches are merged}}
			}
			$\Sigma \gets \varnothing$\;
			$\lambda \gets \Delta$  \tcp*{Start the new $\Sigma$ with the new $\lambda$}
		}
		 \textsf{ENQUEUE}$(\Sigma, \la y_Q, Pa(y_Q) \ra)$\;
		Go to lines~\ref{line:13+}-\ref{line:24+} of Algorithm~\ref{alg:+}
                 }	
    \caption{The revised while loop of the incremental algorithm} \label{alg:+r}
 \end{algorithm}

\section{The decremental case}

In this section we assume that $T$ is an SPT of $G$,  $e_0 = (x_0,y_0)$,  and $\omega(e_0)$ has been decreased to $\omega'(e_0)$. Write $\theta=\omega'(e_0) - \omega(e_0)$. Clearly, $\theta < 0$.  We show how the {\sf Ball-String} algorithm can be adapted to updating SPT in this situation. 

\par
\vspace*{2mm}

\noindent \emph{Notations:} We write $G'$ for the directed graph obtained by decreasing the weight of $e_0$ from $\omega(e_0)$ to $\omega'(e_0)$. For any $u\in V$, we write $\distance_G(u) $ and $\distance_{G'}(u)$ for, respectively, the shortest distance in $G$ and $G'$. Note that, if there is a negative cycle in $G'$, it is possible that $\distance_{G'}(u) = - \infty$. 

We first observe several simple facts.

\begin{lemma} \label{lem:path_length}
For any simple path $\pi$ from $u$ to $v$ in $G'$, if $e_0$ is in $\pi$, then $\length_{G'}(\pi)$ $= \length_{G}(\pi) + \theta$; otherwise, $\length_{G'}(\pi) = \length_{G}(\pi)$, where $\theta=\omega'(e_0)-\omega(e_0)$. 
\end{lemma}
\begin{lemma}\label{lem7}
 If $\distance_G(x_0)+\omega'(e_0) \geq \distance_G(y_0)$, then  $T$ remains an SPT of $G'$.
\end{lemma}
\begin{proof}
If $\distance_G(x_0)+\omega'(e_0) \geq \distance_G(y_0)$, then $e_0=(x_0,y_0)$ is not in $T$ as $\distance_G(x_0) + \omega(e_0)$ is larger than $\distance_G(y_0)$. Therefore, for every vertex $u$, the path from $s$ to $u$ in $T$ has length $\distance_G(u)$ in $G'$. We next show that every other path $\pi$ from $s$ to $u$ has length $\geq \distance_G(u)$. For any simple path $\pi$ from $s$ to $u$ in $G'$, if $e_0$ is in $\pi$, then, since $e_0\not \in \pi[y_0,u]$ and by Lemma~\ref{lem:path_length}, we have $\length_{G'}(\pi[y_0,u])=\length_{G}(\pi[y_0,u])$. Therefore, we have $\length_{G'}(\pi) \geq \distance_G(x_0) + \omega'(x_0,y_0) + \length_{G'}(\pi[y_0,u]) \geq  \distance_G(y_0) + \length_{G'}(\pi[y_0,u])  = \distance_G(y_0) + \length_{G}(\pi[y_0,u]) \geq \distance_G(u)$; if $e_0$ is not in $\pi$, then $\length_{G'}(\pi) = \length_{G}(\pi) \geq \distance_G(u)$. In case $e_0$ appears several times along a path from $s$ to $u$, we can prove inductively that $\length_{G'}(\pi) \geq \distance_G(u)$. Thus $T$ remains an SPT of $G'$. 
\end{proof}

In the following, we characterise when $G'$ has a negative cycle. 

\begin{lemma}
\label{lem:ncycle}
$G'$ has a negative cycle  if and only if there exists a simple path $\pi$ from $y_0$ to $x_0$ in $G$ such that $\length_G(\pi) + \omega'(x_0,y_0)<0$. 
\end{lemma}
\begin{proof}
Suppose $\pi' = \la u_0, u_1, ..., u_{\ell}, u_{\ell+1}=u_0 \ra$ is a negative cycle in $G'$. By Lemma~\ref{lem:path_length} and the fact that $G$ has no negative cycle, $(x_0,y_0)$ is in $\pi'$. Without loss of generality, we assume that $u_0=u_{\ell+1}=x_0$ and $u_1=y_0$. If $(x_0,y_0)$ appears more than once, then, letting $(u_{i}=x_0, u_{i+1}=y_0)$ be the second occurrence, we have two cycles $\pi'[u_0,u_i]$ and $\pi'[u_i,u_{\ell+1}]$. Apparently, at least one cycle is negative. So we may assume that $(x_0,y_0)$ only appears once in $\pi'$.  If $\pi'$ contains other cycles that do not contain $(x_0,y_0)$, we may remove them from $\pi'$ as they have non-negative lengths.   Thus, $\pi=\pi'[u_1,u_{\ell+1}]$ is a simple path from $y_0$ to $x_0$ such that $\length (\pi) + \omega'(x_0,y_0)<0$. The other side is clear as $(x_0,y_0)+\pi$ is a negative cycle in $G'$.
\end{proof} 

As a corollary, we have 
\begin{lemma}\label{lem:negcyc}
If $\distance_G(x_0)+\omega'(x_0,y_0) <  \distance_G(y_0)$ and $y_0$ appears on $\pi_{s,x_0}$ (the path from $s$ to $x_0$ in $T$), then there is a negative cycle at $y_0$. 
\end{lemma}

\begin{proof}
Let $\pi_{s,x_0} = \la u_0=s, u_1, ..., u_i=y_0, ..., u_\ell =x_0 \ra $ be the path from $s$ to $x_0$ in $T$. Since $\pi_{s,x_0}$ is a shortest path in $G$, we know $\pi_{s,x_0}[u_i,u_\ell]$ and $\pi_{s,x_0}[s,u_i=y_0]$ are also shortest paths in $G$. As $G$ has no negative cycle, $\pi_{s,x_0}[u_i,u_\ell]+ (x_0,y_0)$ is a non-negative cycle. Because
\begin{align*}
\distance_G(y_0) & >  \distance_G(x_0)+\omega'(x_0,y_0)\\
&= \length_{G}(\pi_{s,x_0}) + \omega'(x_0,y_0)\\
& = \length_{G}(\pi_{s,x_0}[u_0=s,u_i=y_0]) +  \length_{G}(\pi_{s,x_0}[u_i=y_0, u_\ell=x_0]) +\omega'(x_0,y_0) \\
& =  \distance_G(y_0) +  \length_{G}(\pi_{s,x_0}[u_i, u_\ell]) +\omega'(x_0,y_0),
\end{align*} 
we know $\length_{G}(\pi_{s,x_0}[u_i=y_0, u_\ell=x_0]) +\omega'(x_0,y_0) < 0$ and thus $\pi_{s,x_0}[u_i, u_\ell] + (x_0, y_0)$ is a negative path in $G'$.
\end{proof}
\subsection{Description of the algorithm} 

Suppose $\distance_G(x_0)+\omega'(x_0,y_0) < \distance_G(y_0)$. Starting from $e_0=(x_0,y_0)$, we set 
$$\Delta_T(e_0)=  \distance_G(x_0) + \omega'(x_0,y_0) - \distance_G(y_0)$$ as in Eq.~\eqref{eq:Delta(e)}, which measures the difference between the candidate shortest path $\pi_{s,x_0}+e_0$ in $G'$ and the baseline path $\pi_{s,y_0}$, where $\pi_{s,x_0}$ and $\pi_{s,y_0}$ are the paths from $s$ to $x$ and, respectively, $y_0$ in $T$. Note that $\Delta_T(e_0) < 0$ under our assumption. Furthermore, since $\distance_G(x_0)+\omega(x_0,y_0)  \geq \distance_G(y_0)$, we have $\Delta_T(e_0) \geq \theta $. Thus $\theta \leq \Delta_T(e_0) < 0$. As in Algorithm~\ref{alg:+}, we introduce a queue $Q$ of vertices. Each element of $Q$ has form $\la y, (p, \Delta)\ra$, where $y$ is a vertex, $p$ is a candidate parent of $y$, and $\Delta$ the gain of the candidate path against the baseline path.

 If $y_0$ is not an ancestor of $x_0$ in $T$, then we put $\la y_0, (x_0,\Delta_T(e_0))\ra$ in $Q$ (line~\ref{line:alg-9}) and write $T_1$ for the revised spanning tree.  In the first iteration,  if $e_0$ is not in $T$, then we move the subtree with root $y_0$ in $T$ directly under $x_0$, i.e., changing the parent of $y_0$ as $x_0$ (line~\ref{line:alg-12}). Moreover, we update (in line~\ref{line:alg-16})  the distance from $s$ to each vertex $v$ in $X_1 \equiv\ \downarrow\!y_0$ as $\distance_G(v)+\Delta_T(e_0)$. 

In general, suppose $T_i$ is the spanning tree constructed from $T$ after the $i$-th iteration. We extract from $Q$ the vertex with the smallest $\Delta$ value, written as $\la y_i, (x_i, \Delta)\ra$. By Lemma~\ref{lem:2trees}, we can show $\Delta=\Delta_{T_i}(x_i,y_i)$. We then change the parent of $y_i$ as $x_i$ (line~\ref{line:alg-12}) and obtain a revised spanning tree $T_{i+1}$. The descendants of $y_i$ in $T_i$ are now ready to consolidate (line~\ref{line:alg-16}). Meanwhile, we remove (line~\ref{line:alg-17}) all descendants of $y_i$ in $T_i$ from $Q$, as they have been consolidated and their distances from $s$ cannot be further improved. Write $X_{i+1}$ as the set of descendants of $y_i$ in tree $T_{i}$ and let $Y_{i+1} = Y_i \setminus X_{i+1}$. For any edge $e=(x,y)$ in $G$ with $x\in X_{i+1}$ and $y\in Y_{i+1}$, we define 
$$\Delta_{T_{i+1}}(e) = d_{T_{i+1}}(x) + \omega(x,y) - \distance_G(y)$$ as in Eq.~\eqref{eq:Delta(e)}. Here $d_{T_{i+1}}(x)$ represents the length of the path $\widehat{\pi}_{s,x}$ from $s$ to $x$ in $T_{i+1}$. Note that $x$ has been consolidated in this iteration, and we want to see if we can propagate this change to $y$ through edge $e$.  The value $\Delta_{T_{i+1}}(e)$ represents the \emph{gain} of the new candidate path $\widehat{\pi}_{s,x}+e$ over the baseline path $\pi_{s,y}$. The smaller $\Delta_{T_{i+1}}(e)$ is the better.  
In case $\Delta_{T_{i+1}}(e) \geq 0$, it is impossible to \emph{decrease} the distance of $y$ via $e$, and, thus, not necessary to consider $e$ any more. Otherwise, we put $\la y, (x,\Delta_{T_{i+1}}(e))\ra$ in $Q$ if $y$ is not there or replace the value of $y$ in $Q$ with $\Delta_{T_{i+1}}(e)$  if the latter is smaller (line~\ref{line:alg-25}).  

If $G'$ has no negative cycle, the algorithm will stop when $Q$ is empty.

Recall that a directed graph is \emph{inconsistent} if it has a negative cycle. The procedure of the decremental algorithm (Algorithm~\ref{alg:-}) is  similar to that of the incremental algorithm (Algorithm~\ref{alg:+}). However, it is possible that the edge decrement may result in inconsistency. Lemma~\ref{lem:negcyc} describes such a  simple situation. In general, we can decide if negative cycles exist by checking, as shall be guaranteed by Theorem~\ref{thm4}, if any extracted vertex is an ancestor of $x_0$ (see lines~\ref{line:y_0_inA} and \ref{line:inA}).

 \begin{algorithm}
   \DontPrintSemicolon%
   \SetKwInOut{Input}{Input}%
   \SetKwInOut{Output}{Output}%
   \SetKw{KwSt}{s.t.}%
     \Input{%
       A shortest-path tree $T$ of $G=(V,E)$ and an edge $e_0=(x_0,y_0)$ with a smaller weight $\omega'(e_0) < \omega(e_0)$%
     }%
     \Output{%
      Inconsistent or a shortest-path tree of $G'=(V,E')$, where $E'$ differs from $E$ only in that the weight of $e_0$ is increased from $\omega(e_0)$ to $\omega'(e_0)$.%
     }%
     \BlankLine%
    \textsf{newdist} $\gets \distance_G(x_0) + \omega'(e_0)$\;
     \lIf{\rm \textsf{newdist} $\geq \distance_G(y_0)$}{\Return $T$ \label{line:2-}}\Else{ 
  $A \gets \uparrow\! x_0$ \tcp*{The set of ancestors of $x_0$ in $T$, including $x_0$}    
 $Y \gets V$ \tcp*{The set of unsettled vertices} 
    $\Delta(y_0) \gets {\rm \textsf{newdist}} - \distance_G(y_0)$ \label{line:6-}\;
         \lIf{$y_0\in A$}{\Return \textsf{inconsitent} \label{line:y_0_inA}}\Else{ 
    	\textsf{ENQUEUE}$(Q, \la y_0,(x_0,\Delta(y_0))\ra)$} \label{line:alg-9} 	
\While{$Q\not=\varnothing$}{
      $\la y_Q, (x_Q,\Delta) \ra \gets \textsf{EXTRACTMIN}(Q)$\; 
	\tcp*{When there are various vertices with the same minimum $\Delta$, we select any one with the smallest distance from $s$.} 
         $Pa(y_Q) \gets x_Q$  \label{line:alg-12}    
         \tcp*{Change the parent of $y_Q$ to $x_Q$}
                $X\gets  \downarrow\! y_Q$            \tcp*{$X$ is the set of vertices to be settled in this iteration} 
        $Y \gets Y \setminus X$        \tcp*{$Y$ is the set of unsettled vertices}
        \ForEach{$x\in X$}{
         $D(x) \gets \distance_G(x) + \Delta$  \label{line:alg-16}     \tcp*{All vertices in $X$ are settled.}
                        \lIf{$x \in Q$}{\rm \textsf{REMOVE}$(x,Q)$} \label{line:alg-17} 
        }
	 \ForEach{\rm $(x,y)\in E$ with $x\in X$ and $y \in Y$}{ \label{line:foreach(xy)}
	 \textsf{newdist} $\gets D(x) + \omega(x,y)  $\;
	 \If{ \rm \textsf{newdist} $< \distance_G(y)$ \label{line:20-}} {
	 \lIf{$y\in A$}{\Return \rm \textsf{inconsistent}} \label{line:inA}
	 \Else{
	 $\Delta \gets  \textsf{newdist} - \distance_G(y)$\;
	 \textsf{ENQUEUE}$(Q, \la y, (x,\Delta) \ra)$ \label{line:alg-25}
	 }
	}
	 }
	 }
\Return $T$
}
     \caption{The decremental algorithm} \label{alg:-}
 \end{algorithm}

\subsection{Correctness of the algorithm}

Given an input graph $G=(V,E)$, its SPT $T$, and an edge $e_0=(x_0,y_0)$ with weight decreased from $\omega(e_0)$ to $\omega'(e_0)$, suppose a sequence of $k$ edges $e_i=(x_i,y_i)$ ($i=0,1,...,k-1$) are extracted before the algorithm stops. We note here that, unlike the incremental case, $e_0$ will be put in and extracted from $Q$ first if there is any change of $T$ at all. Write 
\begin{eqnarray} \label{eq:X1}
X_1 &=& \downarrow\!{y_0}\\
\label{eq:X{i+1}}
 X_{i+1} &=& \downarrow\!{y_{i}}\setminus \bigcup_{j=1}^{i} X_j \quad (1\leq i < k)\\ 
 \label{eq:X{k+1}}
 X_{k+1} &=& V\setminus \bigcup_{j=1}^{k} X_j.
 \end{eqnarray} 
 Here $X_{i}$ is understood as the set of vertices that are consolidated in the $i$-th iteration for $i=1,...,k$. For each $1\leq i\leq k-1$, we have $x_i \in \bigcup_{j=1}^i X_j$ and $y_i \in V \setminus \bigcup_{j=1}^i X_j$. Let $T_0=T$ and $T_{i+1}$  $(0 \leq i \leq k-1)$ be the spanning tree modified from $T_i$ after $e_{i}= (x_{i},y_{i})$ is extracted and the parent of $y_i$ has been changed as $x_i$. We define 
 \begin{align}\label{eq:delta{1}}
 \delta_1 &= \Delta_{T_0}(e_{0}) = d_{T_{0}}(x_i) + \omega'(x_0,y_0) -\distance_G(y_0) \\
 \label{eq:delta{i+1}}
 \delta_{i+1} &=  \Delta_{T_i}(e_{i}) = d_{T_{i}}(x_i) + \omega(x_i,y_i) -\distance_G(y_i) \quad\quad (1\leq i \leq k-1).
 \end{align}

\begin{lemma}
\label{lem:yi_path}
 For any $v\in X_{i+1}$, let $\pi$ be the path from $s$ to $v$ in $T_{i+1}$. Then  $d_{T_{i+1}}(v) \equiv \length_{G'}(\pi) = \distance_G(v) + \delta_{i+1} $, where $\delta_{i+1}$ is defined as in \eqref{eq:delta{1}}-\eqref{eq:delta{i+1}}.
\end{lemma}

\begin{proof}
By construction, the result clearly holds for any  $v\in X_1$. Suppose $i>0$. 
Since $d_{T_{i}}(x_i)$ is the length of the path from $v$ to $x_i$ in $T_{i}$ and $x_i$ is the parent of $y_i$ in $T_{i+1}$,   by definition of $\delta_{i+1}$, the path $\widehat{\pi}_{s,y_i}$ from $s$ to $y_i$ in $T_{i+1}$ has length $d_{T_{i}}(x_i) + \omega(x_i,y_i)  =\distance_G(y_{i}) + \delta_{i+1}$. For any $v\in X_{i+1} =\; \downarrow\!{y_i} \setminus \bigcup_{j=1}^i X_j$, the path $\pi$ from $s$ to $v$ in $T_{i+1}$ is the concatenation of $\widehat{\pi}_{s,y_i}$ and $\pi_{s,v}[y_i,v]$, where $\pi_{s,v}$ is the tree path from $s$ to $v$ in $T$. In particular, $\pi_{s,v}[y_i,v]$ is a shortest simple path in both $G$ and $G'$.
Thus the length of $\pi = \widehat{\pi}_{s,y_{i}}+\pi_{s,v}[y_{i},v]$ in $G'$ is $\distance_G(y_i) + \delta_{i+1} + \length_{G}( \pi_{s,v}[y_{i},v]) =  \distance_G(v) + \delta_{i+1}$, as $\distance_G(y_i) = \length_{G}(\pi_{s,v}[s,y_{i}])$ and $\distance_G(v) = \length_{G}( \pi_{s,v})$.
\end{proof}

Then we have the following result, which confirms that Algorithm~\ref{alg:-} always exploits the most profitable paths first.
\begin{lemma}
\label{lem:delta_decrease}
Suppose $\omega'(x_0,y_0) < \omega(x_0,y_0)$ and $e_i=( x_i,y_i )$ $(i=0,1,...,k-1)$ is the sequence of edges that are extracted before the algorithm stops. Then $\theta \leq \delta_1 \leq \delta_2 \leq ... \leq \delta_k < 0$, where $\theta = \omega'(x_0,y_0)-\omega(x_0,y_0)$.
\end{lemma}
\begin{proof}
By definition, $\delta_1  = d_{T_{0}}(x_0) + \omega'(x_0,y_0) - \distance_G(y_0)$. Since $d_{T_0}(x_0)=\distance_G(x_0)$ and, by line 1-3, $\distance_G(x_0) + \omega'(x_0,y_0) < \distance_G(y_0)$, we have $\delta_1 < 0$. Moreover, by  $\omega'(x_0,y_0) =\omega(x_0,y_0) + \theta$, we have $\delta_1= \distance_G(x_0) + \omega(x_0,y_0) - \distance_G(y_0) + \theta $. Because 
$\distance_G(x_0) + \omega(x_0,y_0) \geq  \distance_G(y_0)$, we know $\delta_1 \geq \theta$.

Recall that an edge $e$ is extracted only if the corresponding $\Delta(e)$ is negative. This shows that $\delta_i < 0$ for each $1\leq i\leq k$. We next show that $\delta_j \leq \delta_{j+1}$ for any $1\leq j\leq k-1$. 
By definition, we need only show $\Delta_{T_{j-1}}(e_{j-1}) \leq   \Delta_{T_{j}}(e_{j})$, where $e_{j-1}=(x_{j-1},y_{j-1})$ and $e_j=(x_j,y_j)$ are respectively the $j$-th and $(j+1)$-th extracted edge. There are two subcases. First, if $e_j=(x_{j},y_{j})$ is put in $Q$ or updated in the $j$-th iteration (i.e., after $e_{j-1}$ is extracted), then we have
\begin{align*}
\Delta_{T_{j}}(e_{j}) &= d_{T_{j}}(x_{j}) + \omega(x_{j},y_{j}) - \distance_G(y_{j})\\
			& = \distance_G(x_{j}) + \Delta_{T_{j-1}}(e_{j-1}) + \omega(x_{j},y_{j}) - \distance_G(y_{j}) \\
			& \geq \Delta_{T_{j-1}}(e_{j-1}) 
\end{align*}
because $d_{T_{j}}(x_{j})=\distance_G(x_{j}) + \Delta_{T_{j-1}}(e_{j-1})$ (by Lemma~\ref{lem:yi_path}) and $\distance_G(x_{j}) + \omega(x_{j},y_{j}) \geq \distance_G(y_{j})$. 
Second, suppose $e_j$ is already in $Q$ and not updated in the $j$-th iteration. By the choice of $e_{j-1}$,  $\Delta_{T_{j-1}}(e_{j-1}) \leq   \Delta_{T_{j}}(e_{j})$ 
also holds. This proves $\delta_j \leq \delta_{j+1}$. 
\end{proof}

We next show that if $G'$ has a negative cycle, then Algorithm~\ref{alg:-} will detect this. To show this, we need one additional lemma.

\begin{lemma}
\label{lem:n_iterations}
Suppose no inconsistency is reported by Algorithm~\ref{alg:-}. Then the algorithm stops in at most $n-1$ iterations, where $n$ is the number of vertices in $G$. Assume that $x$
is a vertex which is put in $Q$ in the $\ell$-th iteration. Then either $x$ or one of its ancestor in $T$ is extracted from $Q$ in some later  $\ell'$-th iteration. Moreover,  $x$ will not be put in $Q$ again after the  $\ell'$-th iteration.
\end{lemma}
\begin{proof}
Since no inconsistency is reported in the whole procedure, the algorithm stops only when $Q$ becomes empty.  Assume that there are $k$ vertices being extracted. In each iteration,  after $y_Q$ is extracted,  all its descendants (including $y_Q$ itself) in the current spanning tree are consolidated and removed from $Y$ (line~\ref{line:alg-17}),  the set of unsettled vertices, and will not be put back in $Q$. Thus, $Y$ becomes strictly smaller after each extraction and the procedure stops in at most $n-1$ iterations with an empty $Q$. 

Suppose $x \not=x_0$ is put in $Q$ in iteration $\ell$. Since the algorithm stops with an empty $Q$ after at most $n-1$ iterations, $x$ will be either extracted (as in line~\ref{line:alg-25}) or removed from $Q$  (as in line~\ref{line:alg-17}) in some iteration $\ell'>\ell$. 
\end{proof}

\begin{proposition}\label{prop:2}
If $G'$ has a negative path, then the algorithm will return ``inconsistency" in at most $n-1$ iterations, where $n$  is the number of vertices in $G$.
\end{proposition}
\begin{proof}
First of all, by Lemma~\ref{lem:n_iterations}, the algorithm stops in at most $n-1$ iterations. Suppose $G'$ has a negative path. Then by Lemma~\ref{lem:ncycle} there is a simple path  $\pi = \la u_0=y_0, u_1, ..., u_{\ell-1}, u_{\ell} = x_0 \ra$ in $G$ from $y_0$ to $x_0$ such that $\length_{G}(\pi) + \omega'(x_0,y_0) < 0$. 
Note that $\ell<n$ because $\pi$ is simple.

We prove by reduction to absurdity that the algorithm will return ``inconsistency". Suppose no inconsistency is detected when the algorithm stops.
This means that the algorithm keeps running $k$ iterations until $Q$ is empty for some $1\leq k<n$.  For $0\leq i\leq k-1$, let $e_i=(x_i,y_i)$ be the edge extracted in the $(i+1)$-th iteration. If an edge $(x,y)$ is examined in this iteration and $y$ is put in $Q$  (as in line~\ref{line:foreach(xy)}-\ref{line:alg-25}), then, by Lemma~\ref{lem:n_iterations}, either $y$ or one of its ancestor in $T$ will be extracted in a later iteration.

For each $0\leq i < k$, we assert that: 
\begin{itemize}
\item [A1] There exist $z_i$ and $\widehat{u}_i$ such that $\widehat{u}_i$ is an ancestor of $u_i$\footnote{Recall that we regard $u_i$ also as an ancestor of itself.} and $(z_i, \widehat{u}_i)$ is extracted by the algorithm in iteration, say $\ell_i$. 
\item [A2]  When $(z_i,\widehat{u}_i)$ is extracted in iteration $\ell_i$, we have $\Delta_{T_{ {\ell_i} }}(u_i,u_{i+1}) < 0$, where $T_{ {\ell_i} }$ is the spanning tree constructed in iteration ${ {\ell_i} }$. 
\item [A3] Furthermore, the following inequality holds for $1\leq i < k-1$
\begin{align}
\label{eq:d_t_ell_i}
d_{T_{\ell_{i}}}(u_{i}) \leq d_{T_{\ell_{i-1}}}(u_{i-1}) + \omega(u_{i-1},u_{i}).
\end{align} 
\end{itemize}

We first consider the special case when $i=0$. Let $z_0=x_0$. Then $(x_0,y_0) = (z_0,u_0)$ is extracted in iteration $\ell_0=1$. This is because, otherwise, we shall have $\distance_G(x_0)+ \omega'(x_0,y_0) \geq \distance(y_0)$ and thus $\distance_G(x_0)+ \omega'(x_0,y_0) + \length_{G}(\pi) \geq \distance (y_0) + \length_{G}(\pi) \geq \distance_G(x_0)$, which contradicts the assumption that $\omega'(x_0,y_0) + \length_{G}(\pi) <0$. 
 Let $T_1$ be the spanning tree modified from $T_0=T$ by replacing the parent of $y_0$ as $x_0$ (if $(x_0,y_0)$ is not an edge in $T$) and updating $\omega(x_0,y_0)$ with $\omega'(x_0,y_0)$. We prove that $\Delta_{T_{1}}(y_0,u_1)<0$, i.e., $d_{T_{1}}(y_0) + \omega(y_0,u_1) < \distance_G(u_1)$. Suppose otherwise and $d_{T_1}(y_0) + \omega(y_0,u_1) \geq \distance_G(u_1)$. Recall that $d_{T_1}(y_0) = \distance_G(x_0) + \omega'(x_0,y_0)$. From $\distance_G(x_0) + \omega'(x_0,y_0) + \omega(y_0,u_1) \geq \distance_G(u_1)$ and $\distance_G(u_1) + \length_{G}(\pi[u_1,u_k]) \geq \distance_G(u_{k})= \distance_G(x_0)$, we shall have $\distance_G(x_0) + \omega'(x_0,y_0) + \length_{G}(\pi) \geq  \distance_G(x_0)$, which contradicts the assumption $ \length_{G}(\pi) + \omega'(x_0,y_0) < 0$. This proves $\Delta_{T_1}(y_0,u_1)<0$. 
 
Suppose the above assertions A1-A3 hold for all $j<i$ for some $1\leq i<k-1$. Because $\Delta_{T_{\ell_{i-1}}} (u_{i-1},u_{i}) < 0$, $u_i$ will be put in $Q$ in iteration $\ell_{i-1}$ if not earlier. By Lemma~\ref{lem:n_iterations}, in later iterations, an edge $(z_i,\widehat{u}_i)$ will be extracted such that  $\widehat{u}_i$ is an ancestor of $u_i$ and $(z_i,\widehat{u}_i)$ has a smaller or equal $\Delta$ value as $(u_{i-1},u_i)$. This proves A1 for $i$.

Suppose $(z_i,\widehat{u}_i)$  is extracted  in iteration $\ell_i$ and the spanning tree in this iteration is $T_{\ell_i}$. Since  $\Delta_{T_{\ell_i-1}}(z_i,\widehat{u}_i) \leq \Delta_{T_{\ell_{i-1}}}(u_{i-1},u_i) <0$, we have $d_{T_{\ell_i-1}}(z_i)+\omega(z_i, \widehat{u}_i) \leq \distance_{G}(\widehat{u}_i) + \Delta_{T_{\ell_{i-1}}}(u_{i-1},u_i)$. Furthermore, by $\distance_G(u_i)=\distance_G(\widehat{u}_i)  + \length_G( \pi_{s,u_i}[\widehat{u}_i,u_i] )$ and $\Delta_{T_{\ell_{i-1}}}(u_{i-1},u_i) = d_{T_{\ell_{i-1}}}(u_{i-1}) + \omega(u_{i-1},u_i) - \distance_G(u_i)$, we have 
\begin{align*}
d_{T_{\ell_{i}}}(u_{i}) &= d_{T_{\ell_{i}-1}}(z_{i}) + \omega(z_i, \widehat{u}_i) + \length_G(\pi_{s,u_i}[\widehat{u}_i,u_i]) \\ 
	&\leq \distance_G(\widehat{u}_i) + \Delta_{T_{\ell_{i-1}}}(u_{i-1},u_i) + \length_G(\pi_{s,u_i}[\widehat{u}_i,u_i] ) \\
	& =  \distance_G(\widehat{u}_i)  + d_{T_{\ell_{i-1}}}(u_{i-1}) + \omega(u_{i-1},u_i) - \distance_G(u_i) + \length_G( \pi_{s,u_i}[\widehat{u}_i,u_i] ) \\
	& = (\distance_G(\widehat{u}_i)  + \length_G( \pi_{s,u_i}[\widehat{u}_i,u_i] ) - \distance_G(u_i)) + d_{T_{\ell_{i-1}}}(u_{i-1}) + \omega(u_{i-1},u_i) \\
	& = d_{T_{\ell_{i-1}}}(u_{i-1}) + \omega(u_{i-1},u_i).
\end{align*}
This shows that A3, viz. Eq.~\eqref{eq:d_t_ell_i}, also holds for $i$.

We next show $\Delta_{T_{\ell_i}}(u_i,u_{i+1}) < 0$, i.e., $d_{T_{\ell_i}}(u_i) + \omega(u_i,u_{i+1}) < \distance_G(u_{i+1})$. 
Because A3 holds for all $j< i$, $y_0=u_0$, and $d_{T_1}(y_0) = \distance_G(x_0) +  \omega'(x_0,y_0)$, we have
\begin{align*}
d_{T_{\ell_i}}(u_i) &\leq d_{T_1}(y_0) +\omega(y_0,u_1) + ... + \omega(u_{i-1},u_i) \\
&= \distance_G(x_0) + \omega'(x_0,y_0)+\omega(y_0,u_1) + ... + \omega(u_{i-1},u_i) \\
	& = \distance_G(x_0) + \omega'(x_0,y_0) + \length_{G}(\pi[u_0,u_i]).
\end{align*}
If  $\Delta_{T_{\ell_i}}(u_i,u_{i+1}) \geq 0$, then $\distance_G(u_{i+1}) \leq d_{T_{\ell_i}}(u_i) + \omega(u_i,u_{i+1})$.
Recall $x_0=u_k$ and $\distance_G(x_0)=\distance_G(u_k)\leq \distance_G(u_{i+1})+ \length_{G}(\pi[u_{i+1},u_k])$. We have 
\begin{align*}
\distance_G(x_0) & \leq d_{T_{\ell_i}}(u_i) + \omega(u_i,u_{i+1})+ \length_{G}(\pi[u_{i+1},u_k]) \\
	& \leq \distance_G(x_0) + \omega'(x_0,y_0)+ \length_{G}(\pi[u_0,u_{i+1}]) + \length_{G}(\pi[u_{i+1},u_k]) \\
	& = \distance_G(x_0) + \omega'(x_0,y_0) +\length_{G}(\pi).
\end{align*} 
This contradicts the assumption that $\omega'(x_0,y_0) +\length_{G}(\pi) <0$. Therefore, we must have  $\Delta_{T_{\ell_i}}(u_i,u_{i+1}) < 0$, i.e., A2 also holds for $i$.

In above we have proved that, for any $0\leq i<k$,  we have $\Delta_{T_{\ell_{i}}}(u_{i}, u_{i+1})<0$ when $(z_i,\widehat{u}_i)$ is extracted from $Q$ in the $\ell_{i}$-th iteration. In particular, when  $(z_{k-1}, \widehat{u}_{k-1})$ is extracted from $Q$, we have $\Delta_{T_{\ell_{k-1}}}(u_{k-1}, x_0)<0$, i.e., $d_{T_{\ell_{k-1}}}(u_{k-1}) + \omega(u_{k-1},x_0) < \distance_G(x_0)$.  By Line~\ref{line:inA}, the algorithm shall report ``inconsistency". This contradicts our assumption that the algorithm returns no inconsistency. Therefore, the algorithm does return ``inconsistency". 
\end{proof}

In above, we have shown that if $G'$ has a negative cycle then our algorithm will detect it in at most $n-1$ iterations. In the following we show that if the algorithm stops with a spanning tree $T_{k}$ in the $k$-th iteration (i.e., the last iteration), then $T_{k}$ is an SPT of $G'$. To this end, we need the following result.

\begin{lemma}
\label{lem:c2-q}
Suppose $G'$ has no negative cycle and $\distance_G(x_0) + \omega'(x_0,y_0) < \distance_G(y_0)$. Then $\delta_1=  \distance_{G'}(u) -  \distance_G(u)$ if $u\in\;\downarrow\!{y_0}$, and $\delta_1 \leq  \distance_{G'}(u) - \distance_G(u) \leq 0$ if $u \not\in\;\downarrow\!{y_0}$, where $\delta_1 =  \Delta_{T_1}(x_0,y_0) \equiv  \distance_G(x_0) + \omega'(x_0,y_0) - \distance_G(y_0)$ as in \eqref{eq:delta{i+1}}.
\end{lemma}
\begin{proof}
As in the description of Algorithm~\ref{alg:-}, we write $T_1$ for the spanning tree obtained from $T_0=T$ by setting $x_0$ as the parent of $y_0$ and updating the weight of $(x_0,y_0)$ as $\omega'(x_0,y_0)$.  Let $\theta= \omega'(x_0,y_0) - \omega(x_0,y_0)$. If $e_0$ is in $T$, then $\length_{G'}(\pi_{s,u})=\distance_G(u) + \theta$ for any $u\in \;\downarrow\!{y_0}$ because $e_0$ is also in $\pi_{s,u}$. This shows $\distance_{G'}(u) \leq \distance_G(u) + \theta < \distance_G(u)$. On the other hand, for any $u\in \;\downarrow\!{y_0}$ and any path $\pi$ from $s$ to $u$ in $G'$, we have $\length_{G'}(\pi)=\length_G(\pi)$ or $\length_{G'}(\pi)=\length_G(\pi) + \theta$ by Lemma~\ref{lem:path_length}. Thus we have $\length_{G'}(\pi) \geq \distance_G(u) + \theta$ for any $\pi$ from $s$ to $u$ in $G'$ and, therefore, $\distance_{G'}(u)\geq \distance_G(u) + \theta$. This shows $\distance_{G'}(u) = \distance_G(u) + \theta$.

Now suppose $e_0$ is not in $T$ and $\delta_1 \equiv \distance_G(x_0) + \omega'(x_0,y_0) - \distance_G(y_0)  < 0$. (Note that,  if $e_0$ is in $T$, then $\omega(x_0,y_0)=\distance_G(y_0) - \distance_G(x_0)$ and thus $\delta_1=\theta < 0$.) Since $\pi_{s,x_0}$ and $\pi_{s,u}[y_0,u]$ do not contain $e_0$, by Lemma~\ref{lem:path_length}, we have
\begin{align*}
\length_{G'}(\pi_{s,x_0}) &= \length_{G}(\pi_{s,x_0}) = \distance_G(x_0), \ \mbox{and} \\
\length_{G'}(\pi_{s,u}[y_0,u]) &= \length_{G}(\pi_{s,u}[y_0,u])=\distance_G(u) -\distance_G(y_0).
\end{align*} 
Thus we have 
\begin{align*}
& \ \length_{G'}(\pi_{s,x_0}+(x_0,y_0) + \pi_{s,u}[y_0,u])\\ 
= &  \distance_G(x_0)+ \omega'(x_0,y_0) + \distance_G(u) -\distance_G(y_0)        =  \distance_G(u) + \delta_1.
\end{align*} 
This shows that $\distance_{G'}(u) \leq \length_{G'}(\pi_{s,x_0}+(x_0,y_0) + \pi_{s,u}[y_0,u])  = \distance_G(u) + \delta_1$. 

On the other hand, suppose $\pi$ is a shortest path from $s$ to $u$ in $G'$. Since $\length_{G'}(\pi)=\distance_{G'}(u) < \distance_G(u)$, we know by Lemma~\ref{lem:path_length} that  $e_0$ is in $\pi$. As $G'$ has no negative cycle, we may assume that $\pi$ is a simple path. Clearly, both $\pi[s,x_0]$ and $\pi[y_0,u]$ are also shortest paths in $G'$. Since $e_0$ is in neither $\pi[s,x_0]$ nor $\pi[y_0,u]$, by Lemma~\ref{lem:path_length}, we know $\length_{G'}(\pi[s,x_0]) = \length_{G}(\pi[s,x_0])$  and $ \length_{G'}(\pi[y_0,u])=  \length_{G}(\pi[y_0,u])$. Thus both subpaths are also shortest paths in $G$. This shows that   $\length_{G'}(\pi[s,x_0]) = \distance_G(x_0)$  and $ \length_{G'}(\pi[y_0,u])=  \length_{G}(\pi[y_0,u]) = \distance_G(u) - \distance_G(y_0).$ 
Therefore,
\begin{align*}
	     \length_{G'}(\pi) 
	& = \length_{G'}(\pi[s,x_0])  + \omega'(x_0,y_0) + \length_{G'}(\pi[y_0,u]) \\
	&= \length_{G}(\pi[s,x_0]) + \omega'(x_0,y_0) + \length_{G}(\pi[y_0,u]) \\
	&=  \distance_G(x_0)  + \omega'(x_0,y_0) + \distance_G(u) -\distance_G(y_0) \\
	&= \distance_G(u) + \delta_1.
\end{align*} 
This proves $\delta_1=   \distance_{G'}(u) - \distance_G(u)$.

If $u \not\in \;\downarrow\!{y_0}$, then the path $\pi_{s,u}$ in $T$ does not contain $e_0$ and thus has length $\distance_G(u)$ in $G'$. Thus $\distance_{G'}(u)\leq \distance_G(u)$. Suppose $\pi$ is a shortest path from $s$ to $u$ in $G'$. Since $G'$ has no negative cycle, again, we assume $\pi$ is simple. 
If $\pi$ does not contain $e_0$, then $\distance_{G'}(u)=\length_{G'}(\pi)=\length_{G}(\pi) \geq \distance_G(u)$. Suppose $\pi$ contains $e_0$. Then $\distance_{G'}(u)=\length_{G'}(\pi)=\length_{G}(\pi) + \theta$ and $\pi$ can be decomposed into three sections: $\pi[s,x_0]$, $(x_0,y_0)$, and $\pi[y_0,u]$. Apparently,  both $\pi[s,x_0]$ and $\pi[y_0,u]$ are shortest paths in $G'$ and $G$. Thus, we have  
\begin{align*}
\distance_{G'}(u)=	\length_{G'}(\pi) 
     &=\length_{G'}(\pi[s,x_0])+\omega'(x_0,y_0) + \length_{G'}(\pi[y_0,u])\\
    &=\length_{G}(\pi[s,x_0])+\omega'(x_0,y_0) + \length_{G}(\pi[y_0,u])\\
    &\geq \distance_G(x_0)+\omega'(x_0,y_0) + \distance_G(u)-\distance_G(y_0)\\
    &=\distance_G(u) + \delta_1,
\end{align*} 
where the inequality is because $\distance_G(x_0)=\length_{G}(\pi[s,x_0])$ and
$$\distance_G(u) \leq \length_{G}(\pi_{s,y_0}) + \length_{G}(\pi[y_0,u]) = \distance_G(y_0) + \length_{G}(\pi[y_0,u]).$$ Therefore, if $u \not\in \;\downarrow\!{y_0}$, then $\delta_1 \leq  \distance_{G'}(u) - \distance_G(u) \leq 0 $.
\end{proof}

\begin{proposition}
\label{prop:Vi-}
Suppose $G'$ has no negative cycle and $\distance_G(x_0) + \omega'(x_0,y_0) < \distance_G(y_0)$. 
Let $e_i=(x_i,y_i)$ $(0\leq i \leq  k-1)$ be the sequence of edges that are extracted by the algorithm before termination. 
For any $u\in X_k$, we have $\distance_{G'}(u) =\distance_G(u)$ and, for any $0\leq i \leq k-1$ and any $ u \in X_{i+1}$,  we have 
\begin{align}
\distance_{G'}(u) = d_{T_{i+1}}(u) = \distance_G(u) + \delta_{i+1},
\end{align} 
where 
$ \delta_1$  is as in \eqref{eq:delta{1}} and  $\delta_i$ is as in \eqref{eq:delta{i+1}} for $1\leq i < k$.
\end{proposition}
\begin{proof}
We prove this inductively. Lemma~\ref{lem:c2-q} has shown that $\distance_{G'}(u) = \distance_G(u) + \delta_1$ for every $u\in X_1$. Suppose the result holds for all $j$ smaller than some fixed $i <k$.   We show that it also holds for $j=i$. Given $u\in X_{i+1}$, by Lemma~\ref{lem:yi_path}, we already have a path in $G'$ with length $\distance_G(u) + \delta_{i+1}$. Thus $\distance_{G'}(u) \leq \distance_G(u) + \delta_{i+1}$. Suppose $\pi = \la u_0 =s, u_1, ..., u_\ell=u \ra$ is a shortest path in $G'$ from $s$ to $u$ which is simple and has the fewest number of edges. Since $\length_{G'}(\pi) = \distance_{G'}(u) < \distance_G(u)$, $e_0$ is in $\pi$. Assume that $u_{\ell'}=x_0$, $u_{\ell'+1}=y_0$, and $u_{\ell^*}$ is the last vertex such that $u_{\ell^*}$ is in $X_{j'+1}$ for some $j'<i$. Note that we have $\ell^* \geq \ell'+1$ since $u_{\ell'+1}=y_0 \in X_1$. Moreover, we have $\distance_{G'}(u_{\ell^*}) = \length_{{G'}}(\pi[s,u_{\ell^*}])$ because $\pi$ is a shortest path in $G'$. Now, since $u_{{\ell^*}+1}  \not\in \bigcup_{j=1}^{i} X_j$, it can only be consolidated after $(x_i,y_i)$ is extracted. By Lemma~\ref{lem:delta_decrease}, we have $\delta_{i+1} \leq \delta_{p}$ for any $p>i+1$ and hence $\delta_{i+1}$ is not larger than the value of any vertex in $Q$ when and after $x_{i}$ is extracted. This implies 
\begin{align*}
d_{T_{j'+1}}( u_{\ell^*}) + \omega(u_{\ell^*},u_{{\ell^*}+1}) - \distance_G(u_{{\ell^*}+1}) = \Delta_{T_{j'+1}}(u_{\ell^*}, u_{{\ell^*}+1})  \geq \Delta_{T_{i}}(e_i) = \delta_{i+1}.
\end{align*}
 Because $j'<i$ and, by induction hypothesis, $\distance_{G'}(u_{\ell^*}) = d_{T_{j'+1}}( u_{\ell^*})$, we have 
\begin{align*}
\distance_{G'}(u_{\ell^*}) + \omega(u_{\ell^*}, u_{{\ell^*}+1}) \geq \distance_G(u_{{\ell^*}+1}) + \delta_{i+1}.
\end{align*}
Furthermore, we have 
\begin{align*}
\distance_{G'}(u) &= \length_{G'}(\pi)\\
 &= \distance_{G'}(u_{\ell^*}) + \omega(u_{\ell^*}, u_{{\ell^*}+1}) + \length_{G}(\pi[u_{{\ell^*}+1},u_\ell=u]) \\
& \geq \distance_G(u_{{\ell^*}+1})  + \length_{G}(\pi[u_{{\ell^*}+1},u_\ell=u])  + \delta_{i+1} \\
& \geq \distance_G(u) + \delta_{i+1}.
\end{align*}
Therefore, $\distance_{G'}(u)=\distance_G(u) + \delta_{i+1}$ whenever $u\in X_{i+1}$.

We next show that $\distance_{G'}(u) = \distance_G(u)$ for any $u\in X_{k+1}= V \setminus  \bigcup_{j=1}^k X_{j}$. Since $(x_{k-1},y_{k-1})$ is the last edge that is extracted, for any $x\in \bigcup_{j=1}^k X_{j}$, if $(x,u)$ is an edge in $G'$, then $\Delta_{T_k}(x,u)\geq 0$ holds. This shows that, for any path $\pi$ from $s$ to $u$ in $G'$, we have $\length_{G'}(\pi) \geq \distance_G(u)$. Furthermore, for the path $\pi_{s,u}$ in $T$, we have  $\length_{G'}(\pi_{s,u}) = \distance_G(u)$. This shows $\distance_{G'}(u) = \distance_G(u)$ for any  $u\in X_{k+1}$.
\end{proof}

The above result asserts that if $G'$ has no negative cycle then the algorithm outputs a correct SPT of $G'$. We now show the algorithm is sound and complete.
\begin{theorem}\label{thm4}
Algorithm~\ref{alg:-} is sound and complete. Precisely,  $G'$ has a negative cycle if and only if the algorithm reports `inconsistent'; and, if the algorithm outputs a spanning tree $T'$, then $T'$ is an SPT of $G'$. 
\end{theorem}
\begin{proof}
If $\distance_G(x_0)+\omega'(x_0,y_0) \geq \distance_G(y_0)$, then, by Lemma~\ref{lem7}, $T$ is also an SPT of $G'$ and the algorithm correctly outputs $T$ as specified in line~\ref{line:2-}.

Suppose $\distance_G(x_0)+\omega'(x_0,y_0) < \distance_G(y_0)$. If $G'$ has a negative cycle, then by 
Proposition~\ref{prop:2} the algorithm will return `inconsistent' in at most $n-1$ iterations, where $n$ is the number of vertices in $G$. 

On the other hand, if the algorithm returns `inconsistent' in line~\ref{line:y_0_inA}. Then, by Lemma~\ref{lem:negcyc}, there is a negative cycle in $G'$. 
Similarly, if the algorithm returns `inconsistent' in line~\ref{line:inA}, we assert that there is also a negative cycle in $G'$. Indeed, suppose the algorithm reports this in the $i$-th iteration (i.e., immediately after $e_{i-1}=(x_{i-1},y_{i-1})$ is extracted) for some $0<i \leq k$ and when it examines the edge $(x,y)$ with $x$ settled in the same iteration and $y$ remains unsettled.  Let $\pi$ be the path from $s$ to $x$ in $T_{i}$, the spanning tree obtained by replacing the parent of $y_{i-1}$ with $x_{i-1}$. By $d_{T_{i}}(x)+\omega(x,y)<\distance_{G}(y)$ (line~\ref{line:20-}) and $d_{T_{i}}(x)=\length_{G'}(\pi)$, we know $(x_0,y_0)$ is in $\pi$.  Because $\pi=\pi[s,x_0]+(x_0,y_0)+\pi[y_0,x]$ is a tree path, both $\pi[s,x_0]$ and $\pi[y_0,x]$ are simple paths that do not contain $e_0$. This implies  $\length_{G'}(\pi[s,x_0]) = \length_G(\pi[s,x_0])$ and $\length_{G'}(\pi[y_0,x])=\length_{G}(\pi[y_0,x])$. Similarly, let $\pi_{s,x_0}$ be the tree path from $s$ to $x_0$ in $T$. Then  $\pi_{s,x_0}[y,x_0]$ is also a simple path that does not contain $e_0$. Note that,  before $(x,y)$ is examined, no ancestor of $x_0$ (including $x_0$ itself) is detected to have the property as stated in line~\ref{line:20-}. This suggests that the tree path $\pi_{s,x_0}$ in $T$ is not changed before $(x,y)$ is examined. In particular, we have  $\pi[s,x_0] = \pi_{s,x_0}$.  Now, because $\length_{G'}(\pi)+\omega(x,y)<\distance_{G}(y)=\length_G(\pi_{s,x_0}[s,y])$ and
\begin{align*} 
& \length_{G'}(\pi)+\omega(x,y) \\
					=\ & \length_{G'}(\pi[s,x_0])+ \omega'(x_0,y_0) + \length_{G'}(\pi[y_0,x]) + \omega(x,y)\\
					=\ & \length_G(\pi[s,x_0])+ \omega'(x_0,y_0) + \length_{G}(\pi[y_0,x]) + \omega(x,y) \\
					=\ & \length_G(\pi_{s,x_0})+ \omega'(x_0,y_0) + \length_{G}(\pi[y_0,x]) + \omega(x,y) \\
					=\ & \length_G(\pi_{s,x_0}[s,y])+\length_G(\pi_{s,x_0}[y,x_0])+ \omega'(x_0,y_0) + \length_{G}(\pi[y_0,x]) + \omega(x,y), 
\end{align*}  
we have $\length_G(\pi_{s,x_0}[y,x_0])+ \omega'(x_0,y_0) + \length_{G}(\pi[y_0,x]) + \omega(x,y)  < 0$ and, thus, $\pi_{s,x_0}[y,x_0] +(x_0,y_0)+ \pi[y_0,x] +(x,y)$ is a negative cycle. This shows that if the algorithm reports inconsistency in line~\ref{line:inA} then $G'$ does have a negative cycle.

Lastly, if $G'$ has no negative cycle, then by Proposition~\ref{prop:Vi-} the output spanning tree is indeed an SPT of $G'$.
\end{proof}

We next analyse the computational complexity of the algorithm in terms of the output changes. As in the incremental case, we say a vertex $v$ is \emph{affected} if either its distances from $s$ or its parents in $T$ and $T'$ are different, and say $v$ is \emph{strongly affected} if the parents of $v$ in $T$ and $T'$ are different or when $v=y_0$ is enqueued in $Q$.  In the decremental case, if there are $k$ edges extracted, then the set of affected vertices is $V\setminus X_{k+1} = \bigcup_{i=1}^k X_i$.   Note that every strongly affected vertex has a smaller distance from $s$ in $T'$ than in $T$. This suggests that the set of affected vertices is, unlike the incremental case, independent of the input SPT $T$.

\begin{theorem} \label{thm:comp-}
The time complexity of Algorithm~\ref{alg:-} is $O(m_0+n_0 \log n_0)$ if $Q$ is implemented as a relaxed heap, where $n_0$ is the number of affected vertices, $m_0$ is the number of edges in $G$ whose tails are affected vertices. 
\end{theorem}
\begin{proof}
The proof is analogous to that for Algorithm~\ref{alg:+}. 
\end{proof}

Similar to the incremental case, we can show that the output spanning tree $T'$ can also be modified in $O(n_s)$ time into an SPT of $G'$ which has minimal edge changes, where $n_s$ is the number of strongly affected vertices.
\begin{theorem} \label{thm:minchange-}
Suppose that $G$ has no 0-cycles and $G'$ has no negative cycles. Let $T$ be the SPT of $G$ in the input of Algorithm~\ref{alg:-} and $T'$ any SPT of $G'$. Constructing  $T^*$ by merging linked $T'$-branches, then $T^*$ has minimal edge changes to $T$.
\end{theorem}  
\begin{proof}
The proof is analogous to that for Algorithm~\ref{alg:+}. 
\end{proof}
As the incremental case, the revised SPT $T^*$ can be obtained by an analogous auxiliary simple procedure as described in Algorithm~\ref{alg:+r}, which introduces only $O(n_s)$ extra operations. Here we need the result that $\delta_i\leq \delta_{i+1}$ for $1\leq i\leq k$, obtained in Lemma~\ref{lem:delta_decrease}, which guarantees that the `gain' in searching shorter paths becomes smaller and smaller.

\section{Conclusion}
In this paper we showed that how the well-known {\sf Ball-String} algorithm can be adapted to solve the dynamical SPT problem for directed graphs with arbitrary real weights. While the adapted incremental algorithm is almost the same as the original one in \cite{NarvaezST01}, the adapted decremental algorithm introduces procedures for checking if negative cycles exist. We rigorously proved the correctness of the adapted algorithms and analysed their computational complexities in terms of the output changes. Unlike the algorithms of Ramalingam and Reps \cite{RamaR96} and Frigioni et al. \cite{FrigioniMN03}, the semi-dynamic {\sf Ball-String} algorithms discussed in this paper can deal with directed graphs with zero-length cycles without extra efforts and are more \emph{efficient} because they extract fewer vertices from the priority queue and always consolidate the whole branch when a vertex is extracted. Furthermore, we showed by an example that the output SPT of the {\sf Ball-String} algorithm does not necessarily have minimal edge changes and then proposed a very general method for transforming any SPT of the revised graph into one with minimal edge changes in time linear in the number of extractions.

Note that we only consider the case when only one edge is revised in this paper. When a small number of edges are revised simultaneously, we may compute the SPT of the revised graph step by step, where in each step we consider only one edge update.  It is worth noting that,  when multiple edges are decreased simultaneously, the {\sf Ball-String} algorithm does not always output the correct SPT with minimal edge changes even when only positive weights present \cite{NarvaezST01}. If this is the case, we can transform the output SPT into an SPT with minimal edge changes using the method described in our Algorithm~\ref{alg:+r}.

\bibliographystyle{plain}

\bibliography{sssp}

\end{document}